\documentclass[11pt,draftclsnofoot]{IEEEtran}
\onecolumn

\usepackage{amsmath, amssymb,amsthm,cite}
\usepackage[dvips]{graphics}
\usepackage{graphicx}
\usepackage{psfrag}
\usepackage{cancel}
\usepackage{bbm}
\usepackage{dblfloatfix}
\usepackage{mathtools}
\usepackage{algpseudocode}
\usepackage{algorithm}
\usepackage{slashbox}
\usepackage{hyperref}
\usepackage{color}
\usepackage{pifont}

\newtheorem*{rep@theorem}{\rep@title}
\newcommand{\newreptheorem}[2]{%
\newenvironment{rep#1}[1]{%
 \def\rep@title{#2 \ref{##1}}%
 \begin{rep@theorem}}%
 {\end{rep@theorem}}}

\DeclareMathOperator*{\argmax}{\arg\max}

\allowdisplaybreaks

\newcommand{\RNum}[1]{\uppercase\expandafter{\romannumeral #1\relax}}
\newtheorem{lemma}{Lemma}
\newreptheorem{lemma}{Lemma}
\newtheorem{theorem}{Theorem}

\theoremstyle{definition}

\newtheorem{definition}{Definition}

\makeatletter
\def\blfootnote{\gdef\@thefnmark{}\@footnotetext}
\makeatother

\newcommand{\xmark}{\ding{55}}%
\makeatletter
\def\blfootnote{\gdef\@thefnmark{}\@footnotetext}
\makeatother

\begin{document}
\title {Feedback Capacity and Coding for the $(0,k)$-RLL Input-Constrained BEC}
\author{Ori Peled, Oron Sabag and Haim H. Permuter}
\maketitle

\begin{abstract}
  The input-constrained binary erasure channel (BEC) with strictly causal feedback is studied. The channel input sequence must satisfy the $(0,k)$-runlength limited (RLL) constraint, i.e., no more than $k$ consecutive `$0$'s are allowed. The feedback capacity of this channel is derived for all $k\geq 1$, and is given by $$C^\mathrm{fb}_{(0,k)}(\varepsilon) = \max\frac{\overline{\varepsilon}H_2(\delta_0)+\sum_{i=1}^{k-1}\left(\overline{\varepsilon}^{i+1}H_2(\delta_i)\prod_{m=0}^{i-1}\delta_m\right)}{1+\sum_{i=0}^{k-1}\left(\overline{\varepsilon}^{i+1}
   \prod_{m=0}^{i}\delta_m\right)},$$ where $\varepsilon$ is the erasure probability, $\overline{\varepsilon}=1-\varepsilon$ and $H_2(\cdot)$ is the binary entropy function. The maximization is only over $\delta_{k-1}$, while the parameters $\delta_i$ for $i\leq k-2$ are straightforward functions of $\delta_{k-1}$. The lower bound is obtained by constructing a simple coding for all $k\geq1$. It is shown that the feedback capacity can be achieved using zero-error, variable length coding. For the converse, an upper bound on the non-causal setting, where the erasure is available to the encoder just prior to the transmission, is derived. This upper bound coincides with the lower bound and concludes the search for both the feedback capacity and the non-causal capacity. As a result, non-causal knowledge of the erasures at the encoder does not increase the feedback capacity for the $(0,k)$-RLL input-constrained BEC. This property does not hold in general: the $(2,\infty)$-RLL input-constrained BEC, where every `$1$' is followed by at least two `$0$'s, is used to show that the feedback capacity can be strictly greater than the non-causal capacity.
   %Finally, the general $(d,k)$-RLL constraint is discussed and a solution is provided for the $(1,2)$-RLL constraint, where every '1' must be followed by a '0' and no more than two consecutive '0's are allowed.
\end{abstract}


\section{Introduction}\label{sec:intro}
\blfootnote{The work  was supported in part by the European Research Council under the European Union’s Seventh Framework Programme (FP7/2007-2013)/ERC grant agreement no.337752, and the ISF research grant 818/17. Part of this work was presented at the 2017 IEEE Int. Symposium on Information Theory (ISIT 2017), Aachen, Germany \cite{Peled_ISIT}.O. Peled, O. Sabag and H. H. Permuter are with the department of Electrical and Computer Engineering, Ben-Gurion University of the Negev, Beer-Sheva, Israel (oripe@post.bgu.ac.il, oronsa@post.bgu.ac.il, haimp@bgu.ac.il).}
The physical limitations of the hardware used in recording and communication systems cause some digital sequences to be more prone to errors than others. This elicits the need to ensure that such sequences will not be recorded or transmitted. Constrained coding is a method that enables such systems to encode arbitrary data sequences into sequences that abide by the imposed restrictions \cite{Marcus98}. In the classical constrained coding setting, it is assumed that the transmission is noiseless if the transmitted sequence satisfies the imposed constraint. In this paper, however, we consider a transmission of constrained sequences where the transmission is over a noisy channel, the binary erasure channel (BEC) (Fig. \ref{figure:BEC}).

Run-length limited (RLL) constraints are common in magnetic and optical recording standards, where the run length of consecutive `$0$'s should be limited between $d$ and $k$ $(d<k)$. A $(d,k)$-RLL constrained binary sequence must satisfy two restrictions:
\begin{enumerate}
  \item at least $d$ `$0$'s must follow each `$1$'.
  \item no more than $k$ consecutive `$0$'s are allowed.
\end{enumerate}
The first restriction ensures that the frequency of transitions, i.e., $1 \to 0$ or $0 \to 1$, will not be too high. This is necessary in systems where the sequence is conveyed over band-limited channels. Timing is commonly recovered with a phase-locked loop (PLL) that adjusts the phase of the detection instant according to the observed transition of the received waveform. The second restriction guarantees that the PLL does not fall out of synchronization with the waveform \cite{Immink90runlength-limitedsequences,Marcus98}.

Two important families of the RLL constraint are the $(d,\infty)$-RLL and $(0,k)$-RLL. These constraints might seem symmetric in some sense, but indeed, may greatly differ in their behavior, see e.g., \cite{dk_1,dk_2}. Therefore, when dealing with RLL constraints, it is common to tackle each of these families separately before approaching the general $(d,k)$-RLL. In this paper, we adopt this approach and show that the $(0,k)$-RLL problem is solvable, while the same problem with a $(d,\infty)$-RLL constraint is a great deal more challenging.

The model studied in this paper is a BEC (Fig. \ref{figure:BEC}), in which the input sequences must satisfy the $(0,k)$-RLL constraint. Two cases of this model are investigated, based on the information that is available to the encoder. In the first case, described in Fig. \ref{figures:fb_channel}, the encoder has access to all past outputs via a noiseless feedback link. In the second case, described in Fig. \ref{figures:nc_channel}, the encoder has non-causal access to the erasure that is about to occur, that is, the encoder knows in advance whether the BEC behaves like a clean channel or not. From an operational point of view, the capacity of the non-causal case must be greater than the feedback case due to the additional information that the encoder has.

\begin{figure}
  \centering
  \psfrag{A}[c][][.8]{$1-\varepsilon$}
  \psfrag{B}[c][r][.8]{$\varepsilon$}
  \psfrag{C}[c][][.8]{$X$}
  \psfrag{D}[c][][.8]{$Y$}
  \psfrag{E}[c][][.8]{?}
  \includegraphics[scale=0.85]{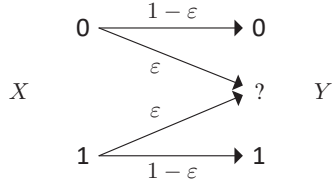}
  \caption{Binary erasure channel with erasure parameter $\varepsilon$.}\label{figure:BEC}
\end{figure}
\begin{figure}
  \centering
  \psfrag{A}[c][][.85]{\shortstack{$(0,k)$-RLL \\ Constrained \\ Encoder}}
  \psfrag{B}[c][][.9]{BEC}
  \psfrag{C}[c][][.9]{Decoder}
  \psfrag{D}[c][][.9]{Unit Delay}
  \psfrag{X}[c][][.9]{$M$}
  \psfrag{Y}[c][][.9]{$Y_{i-1}$}
  \psfrag{Z}[c][][.9]{$X_i$}
  \psfrag{W}[c][][.9]{$Y_i$}
  \psfrag{T}[c][][.9]{$\hat{M}$}
  \includegraphics[scale=0.75]{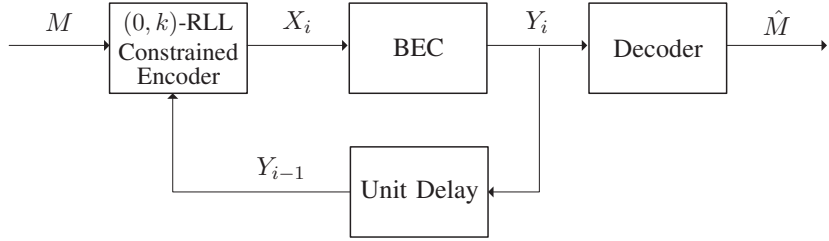}
  \caption{Input constrained BEC with strictly causal feedback. The channel input $X_i$ is a function of the message $M$ and of the channel output history $Y^{i-1}$.}\label{figures:fb_channel}
\end{figure}
\begin{figure}
  \centering
  \psfrag{A}[c][][.85]{\shortstack{$(0,k)$-RLL \\ Constrained \\ Encoder}}
  \psfrag{B}[c][][.79]{\shortstack{BEC \\ $Y_i=f(X_i,\theta_i)$}}
  \psfrag{C}[c][][.9]{Decoder}
  \psfrag{X}[c][][.9]{$M$}
  \psfrag{Y}[c][][.9]{$\theta^i$}
  \psfrag{Z}[c][][.9]{$X_i$}
  \psfrag{W}[c][][.9]{$Y_i$}
  \psfrag{T}[c][][.9]{$\hat{M}$}
  \includegraphics[scale=0.75]{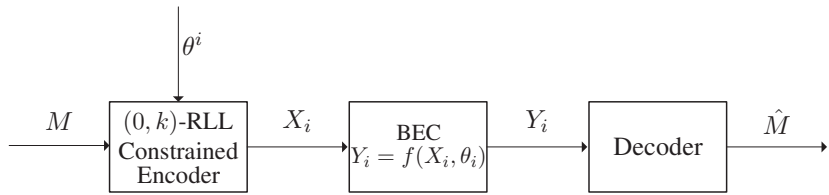}
  \caption{Input-constrained BEC with non-causal knowledge of the erasures. The encoder has access both to the message $M$ and to $\theta^i$ which model the erasure.}\label{figures:nc_channel}
\end{figure}

We show that the feedback capacity of the $(0,k)$-RLL input-constrained BEC is:
\begin{equation}\label{eq:k_capacity_in_intro}
    C^{\mathrm{fb}}_{(0,k)}(\varepsilon) = \max_{0\leq\delta_{k-1}\leq \frac{1}{2}} \frac{\overline{\varepsilon}H_2(\delta_0)+\displaystyle\sum_{i=1}^{k-1}\left(\overline{\varepsilon}^{i+1}H_2(\delta_i)\prod_{m=0}^{i-1}\delta_m\right)}{1+\displaystyle\sum_{i=0}^{k-1}\left(\overline{\varepsilon}^{i+1}
   \prod_{m=0}^{i}\delta_m\right)},
\end{equation}
for all $\varepsilon\in[0,1]$ and $k\geq1$, where $\delta_0,\dots,\delta_{k-2}$ are simple functions of $\delta_{k-1}$, given in Eq. \eqref{eq:delta's}, below. Surprisingly, we are also able show that the non-causal knowledge of the channel erasure does not increase the feedback capacity, so that \eqref{eq:k_capacity_in_intro} is the non-causal capacity as well.

This work generalizes the results in \cite{Sabag_BEC}, where the feedback capacity of the $(1,\infty)$-RLL input-constrained BEC was calculated\footnote{The $(1,\infty)$-RLL constraint is equivalent to the $(0,1)$-RLL constraint by swapping `$0$'s and `$1$'s}. In \cite{Sabag_BEC} and other works, \cite{TatikondaMitter_IT09,Chen05,Yang05,PermuterCuffVanRoyWeissman08,trapdoor_generalized,Ising_channel,generalized_Ising,sabag_BIBO}, the capacity was derived by formulating it as a dynamic programming (DP) problem and then solving the corresponding Bellman equation. In all past works, the DP state was an element of the $1$-simplex, an essential property in the solution of the Bellman equation. However, the DP state in our case is an element of the $k$-simplex. This makes the approach of solving the Bellman equation intractable and different methods are required.

%Specifically, for a broad family of channels with memory, including our setting, it is possible to have an equivalent DP problem in which the optimal average reward is equal to the feedback capacity \cite{TatikondaMitter_IT09,Chen05,Yang05,PermuterCuffVanRoyWeissman08}. \textcolor{red}{Then, it is sometimes possible to solve the fixed-point Bellman equation that serves as a sufficient condition for concluding the optimality of a certain expression, like has been done in \cite{PermuterCuffVanRoyWeissman08,trapdoor_generalized,Ising_channel,generalized_Ising,Sabag_BEC,sabag_BIBO}.}

To circumvent this difficulty, we use alternative techniques to solve the capacity of our problems. The upper bound follows from standard converse techniques for the non-causal model, where the encoder knows the erasure ahead of time. This upper bound is trivially an upper bound also for the feedback model, since non causal knowledge might increase the capacity only. Then, we construct a simple coding scheme for the feedback setting, inspired by the posterior matching principle \cite{horstein_original,shayevitz_posterior_mathcing,Li_elgamal_matching,shayevitz_simple}. The coding scheme enables both the encoder and the decoder to systematically reduce the size of the set of possible messages to a single message, which is then declared by the decoder as the correct message. An analysis of the achieved rate reveals an expression that is similar to the upper bound. The equivalence of these bounds is finally derived, and this concludes both the feedback capacity and the non-causal capacity for our setting.

The remainder of the paper is organized as follows: Section \ref{sec:Notation&Problem_Definition} includes the notations we use and the problem definition. Section \ref{sec:main_results} contains the main results of this paper. In Section \ref{sec:coding_scheme} we present the coding scheme and its rate analysis. Section \ref{sec:upper_bound} includes an upper bound of the capacity. In Section \ref{sec:(2,infty)-RLL} we discuss the $(2,\infty)$-RLL input constraint, as an example where the non-causal capacity is strictly greater than the feedback capacity. Section \ref{sec:(1,2)-RLL} presents the feedback capacity of the $(1,2)$-RLL BEC, as an example for possible future avenues of research. Finally, the appendices contain proofs of several lemmas used throughout the paper.

\section{Notations and Problem Definition}\label{sec:Notation&Problem_Definition}
\subsection{Notations}
Random variables are denoted using a capital letter $X$. Lower-case letters $x$ are used to denote realizations of random variables. Calligraphic letters $\mathcal{X}$ denote sets. The notation $X^n$ means the $n$-tuple $\left(X_1,\ldots,X_n\right)$ and $x^n$ is a realization of such a vector of random variables. For a real number $\alpha\in\left[0,1\right]$, we define $\overline{\alpha} \coloneqq 1-\alpha$. The binary entropy function is denoted by $H_2(\alpha)=-\alpha\log_2\alpha-\overline{\alpha}\log_2\overline{\alpha}$ for $\alpha\in\left[0,1\right]$.

\subsection{Problem Definition}
The BEC (Fig. \ref{figure:BEC}) is memoryless, that is $p(y_i \mid x^i,y^{i-1}) = p(y_i \mid x_i)\ \forall i$, and can be represented by:
 \begin{equation*}
   y_i = \begin{cases}
           x_i, & \mbox{if } \theta_i=\checkmark \\
           ?, & \mbox{if } \theta_i=\text{\xmark}
         \end{cases},
 \end{equation*}
 where $\theta^n$ is an i.i.d. process with $\theta_i\sim Ber(\varepsilon)$. A message $M$ is drawn uniformly from $\{1,2,\ldots,2^{nR}\}$ and is available to the encoder. We define two models, based on the additional information that is available to the encoder: in the first model, at time $i$, the encoder has knowledge of past channel outputs $y^{i-1}$ via a noiseless feedback link (Fig. \ref{figures:fb_channel}); in the second model, at time $i$, the encoder has non-causal access to $\theta_i$ (Fig. \ref{figures:nc_channel}). In both cases, the transmission is over a BEC.

\begin{figure}
  \centering
  \psfrag{A}[c][][.8]{$S=0$}
  \psfrag{B}[c][][.8]{$S=1$}
  \psfrag{c}[c][][.7]{$S={k-1}$}
  \psfrag{D}[c][][.8]{$S=k$}
  \psfrag{E}[c][][.8]{$1$}
  \psfrag{F}[c][][.8]{$0$}
  \includegraphics[scale=0.85]{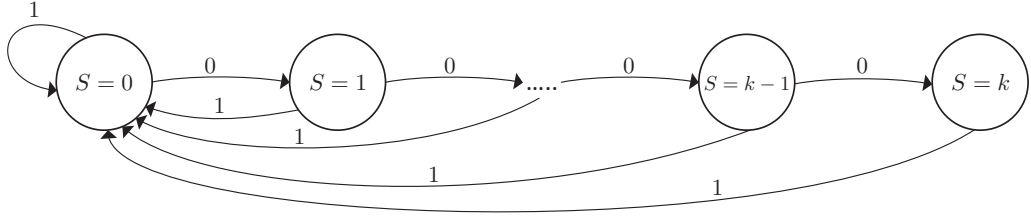}
  \caption{State diagram describing all sequences that can be generated while satisfying the $(0,k)$-RLL constraint. Note that after k consecutive `$0$'s the node $S=k$ is reached, which implies that the next bit is necessarily `$1$'.}\label{figures:(0,k)_RLL}
\end{figure}

The encoder must produce sequences that comply with the $(0,k)$-RLL input constraint. This constraint can be described graphically (Fig. \ref{figures:(0,k)_RLL}), where all walks along the directed edges of the graph do not contain the forbidden string. Note that the node $S=k$ has only one outgoing edge, labeled `$1$', which implies that after $k$ consecutive `$0$'s, the next bit must be a `$1$'. The constrained encoder and the decoder operations are made precise by the following code definitions.

\begin{definition}
  A $(n,2^{nR},(0,k))$ \textit{code} for an input-constrained BEC is composed of encoding and decoding functions.
The encoding functions for the first model (with feedback) are:
  \begin{equation}\label{eq:fb_enc}
    f_i:\{1,\ldots,2^{nR}\}\times\mathcal{Y}^{i-1} \rightarrow \mathcal{X}_i, \ i=1,\ldots,n ,
  \end{equation}
  satisfying $f_i\left(m,y^{i-1}\right) = 1$ if $\left( f_{i-1}\left(m,y^{i-2}\right),\ldots,f_{i-k}\left(m,y^{i-k-1}\right) \right)=(0,\ldots,0)$ for all $\left( m,y^{i-1} \right)$.
  For the non-causal model the encoding functions are defined by:
  \begin{equation}\label{eq:nc_enc}
    g_i:\{1,\ldots,2^{nR}\}\times\{\checkmark,\text{\xmark}\}^i\rightarrow\mathcal{X}_i, \ i=1,\ldots,n ,
  \end{equation}
  satisfying $g_i\left(m,\theta^i\right) = 1$ if $\left( g_{i-1}\left(m,\theta^{i-1}\right),\ldots,g_{i-k}\left(m,\theta^{i-k}\right) \right)=(0,\ldots,0)$ for all $\left( m,\theta^i \right)$. The decoding function for both models is defined by:
  \begin{equation*}
    \Psi : \mathcal{Y}^n \rightarrow \{1,\ldots,2^{nR}\}.
  \end{equation*}
Without loss of generality, we assume that $x_0=1$, so that the initial state is $s_0=0$.

The \textit{average probability of error} for a code is defined as $P_e^{(n)}=Pr\left(M\neq\Psi(Y^n)\right)$. A rate $R$ is said to be $(0,k)$-\textit{achievable} if there exists a sequence of $(n,2^{nR},(0,k))$ codes such that $\lim_{n\to\infty}P_e^{(n)}=0$. The \textit{capacity} is defined to be the supremum over all $(0,k)$-achievable rates and is a function of $k$ and the erasure probability $\varepsilon$. Denote by $C^{\mathrm{fb}}_{(0,k)}(\varepsilon)$ the capacity of the feedback model and $C^{\mathrm{nc}}_{(0,k)}(\varepsilon)$ that of the non-causal model. Since $y^{i-1}$ is computable from $\theta^{i-1}$ and $M$, we have the relation $C^{\mathrm{nc}}_{(0,k)}(\varepsilon)\geq C^{\mathrm{fb}}_{(0,k)}(\varepsilon)$, for all $k\geq1$, $\varepsilon\in[0,1]$.

\end{definition}

\section{Main Results}\label{sec:main_results}
In this section we present the main results, including the feedback capacity and the non-causal capacity for the BEC with $(0,k)$-RLL input constraints. We then explain the methodology used to prove the results. The following theorem constitutes our main results regarding the feedback capacity and the capacity achieving coding scheme. Define the function:
\begin{equation}\label{eq:capacity_expression}
  R_\varepsilon\left(\delta_0, \ldots ,\delta_{k-1}\right) =  \frac{\overline{\varepsilon}H_2(\delta_0)+\displaystyle\sum_{i=1}^{k-1}\left(\overline{\varepsilon}^{i+1}H_2(\delta_i)\prod_{m=0}^{i-1}\delta_m\right)}{1+\displaystyle\sum_{i=0}^{k-1}\left(\overline{\varepsilon}^{i+1}
   \prod_{m=0}^{i}\delta_m\right)} ,
\end{equation}
where $\delta_i$ takes values in $[0,1]$ for $i=0,\ldots,k-1$.

\begin{theorem}\label{thm:capacity}
The feedback capacity of the $(0,k)$-RLL input-constrained BEC with feedback is:
\begin{equation}\label{eq:k_capacity}
    C^{\mathrm{fb}}_{(0,k)}(\varepsilon) = \max_{0\leq\delta_{k-1}\leq \frac{1}{2}} R_\varepsilon\left(\delta_0, \ldots ,\delta_{k-1}\right),
\end{equation}
where $\delta_0,\dots,\delta_{k-2}$ are functions of $\delta_{k-1}$ and can be calculated recursively using:
\begin{equation}\label{eq:delta's}
    \delta_j = \frac{\delta_{j+1}}{\delta_{j+1}+\overline{\delta}_{j+1}\left( \frac{\overline{\delta}_{j+1}}{\overline{\delta}_{j+2}} \right)^{\overline{\varepsilon}}} \quad j=0,1,\ldots,k-2 ,
\end{equation}
with $\overline{\delta}_k \coloneqq 1$.
  In addition, there exists a simple coding scheme that achieves the capacity given in \eqref{eq:k_capacity}.

\end{theorem}

Fig \ref{figures:various_k} presents graphs of the feedback capacity as a function of $\varepsilon$ for several values of $k$. The capacity is a decreasing function of $\varepsilon$, and an increasing function of $k$. For $\varepsilon=0$, the channel is noiseless and so the capacity is that of the corresponding constraint. For example, $C^{\mathrm{fb}}_{(0,1)}(0) = \log_2(\phi)$, where $\phi = \frac{1+\sqrt{5}}{2}$ is the golden ratio, which is known to be the capacity of sequences that do not contain two consecutive `$0$'s. For $\varepsilon = 1$, the output is constant so we have $C^{\mathrm{fb}}_{(0,k)}(1) = 0$ for all $k$. As $k$ increases the constraint becomes more lenient and the capacity approaches $1-\varepsilon$, which is the capacity of the unconstrained BEC.

\begin{figure}[h!]
  \centering
  \psfrag{A}[c][][.8]{Erasure probability $\varepsilon$}
  \psfrag{B}[b][][.9]{The feedback capacity $C^{\mathrm{fb}}_{(0,k)}(\varepsilon)$}
  \psfrag{C}[c][][.8]{Feedback capacity}
  \includegraphics[scale=0.50]{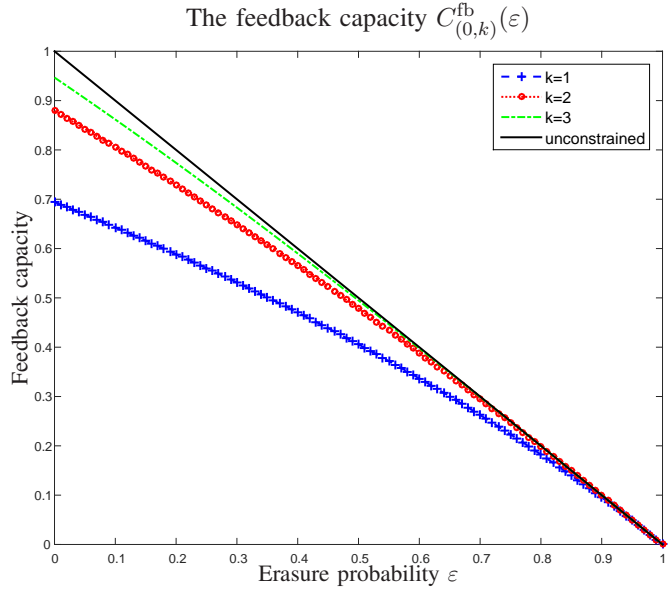}
  \caption{Feedback capacity as a function of $\varepsilon$ for several values of $k$ and the unconstrained capacity. As $k$ increases, the performance approaches that of the unconstrained channel.}\label{figures:various_k}
\end{figure}

Theorem \ref{thm:capacity} guarantees that even though the function we aim to maximize is a function of $k$ variables, to calculate the capacity, one needs only to perform a maximization over $\delta_{k-1}$. For any $\delta_{k-1}\in [0,1]$, the values of all other variables can be calculated by utilizing the set of equations given in \eqref{eq:delta's}.

Our proposed coding scheme has $k$ degrees of freedom, represented by $\delta_0,\ldots,\delta_{k-1}$. For this reason, it is rather surprising that the feedback capacity is a simple optimization problem of one variable for all $k\geq1$. Indeed, the relaxation of the optimization domain shows that optimizing over the k-tuple and the optimization in \eqref{eq:k_capacity} and \eqref{eq:delta's} are equivalent. In addition we also prove the following:
\begin{theorem}\label{thm:fb_capacity=nc_capacity}
  Non-causal knowledge of the erasures does not increase the feedback capacity, that is $\forall k\geq1,\varepsilon\in[0,1]$:
  \begin{equation*}
    C^{\mathrm{fb}}_{(0,k)}(\varepsilon) = C^{\mathrm{nc}}_{(0,k)}(\varepsilon).
  \end{equation*}
\end{theorem}

It is tempting to conjecture that this property holds for the general $(d,k)$-RLL constrained BEC, but we will provide a counterexample in Section \ref{sec:(2,infty)-RLL}. Theorems \ref{thm:capacity} and \ref{thm:fb_capacity=nc_capacity} both generalize parallel results shown in \cite{Sabag_BEC}, where the special case of $k=1$ was calculated using different techniques.
The following inequalities are the main steps required to prove Theorems \ref{thm:capacity} and \ref{thm:fb_capacity=nc_capacity}:
\begin{equation}\label{eq:inequalities}
  \max_{0\leq\delta_0,\ldots,\delta_{k-1}\leq \frac{1}{2}} R_\varepsilon\left(\delta_0, \ldots ,\delta_{k-1}\right) \stackrel{(a)}\leq
  C^{\mathrm{fb}}_{(0,k)}(\varepsilon) \stackrel{(b)}\leq C^{\mathrm{nc}}_{(0,k)}(\varepsilon) \stackrel{(c)}\leq \max_{0\leq\delta_0,\ldots,\delta_{k-1}\leq 1} R_\varepsilon\left(\delta_0,\ldots,\delta_{k-1}\right) ,
\end{equation}
where
\begin{itemize}
  \item Inequality (a) follows from the coding scheme that is presented in Algorithm \ref{alg:Coding_Scheme}. Specifically, it is shown that $R_\epsilon(\delta_0,\ldots,\delta_{k-1})$ is achievable for any choice of $\delta_i \leq 0.5, \ i=0,\ldots,k-1$.
  \item Inequality (b) follows from the operational definitions of the code in Section \ref{sec:Notation&Problem_Definition}.
  \item Inequality (c) follows from standard converse techniques for the non-causal setting.
\end{itemize}

The next lemma shows that the maximal value of $R_\varepsilon(\delta_0,\ldots,\delta_{k-1})$ remains the same whether the maximization domain is $0\leq\delta_0,\ldots,\delta_{k-1}\leq \frac{1}{2}$ or $0\leq\delta_0,\ldots,\delta_{k-1}\leq 1$. Thus, the chain of inequalities is actually a chain of equalities.
\begin{lemma}\label{lemma:main_lemma}
For all $\varepsilon\in [0,1]$ and $k\geq1$,
  \begin{equation*}
    \max_{0\leq\delta_0,\ldots,\delta_{k-1}\leq \frac{1}{2}} R_\varepsilon\left(\delta_0, \ldots ,\delta_{k-1}\right) =
    \max_{0\leq\delta_0,\ldots,\delta_{k-1}\leq 1} R_\varepsilon\left(\delta_0,\ldots,\delta_{k-1}\right).
  \end{equation*}
  Moreover, the k-tuple $\argmax_{0\leq\delta_0,\ldots,\delta_{k-1}\leq \frac{1}{2}} R_\varepsilon\left(\delta_0, \ldots ,\delta_{k-1}\right)$ satisfies Eq. \eqref{eq:delta's}.
\end{lemma}
Theorems \ref{thm:capacity} and \ref{thm:fb_capacity=nc_capacity} are concluded from the inequalities chain \eqref{eq:inequalities} and Lemma \ref{lemma:main_lemma}, which is proved in Appendix \ref{sec:equality_of_bounds}

 \section{Optimal Coding Scheme for the Input-Constrained BEC with Feedback}\label{sec:coding_scheme}
 In this section, we introduce the idea behind the optimal coding scheme, as well as the coding itself, which is presented in Algorithm \ref{alg:Coding_Scheme}. We then prove that the scheme is feasible, meaning that the generated input sequence does not violate the input constraint, and, finally, calculate its rate.
 \subsection{Coding Scheme}
 Before presenting the coding scheme itself, we discuss the basic ideas in accordance with which the scheme operates. The coding scheme is a mechanism that allows the encoder to transmit a message $m\in\mathcal{M}=\{1,2,\ldots,2^{nR}\}$ to the decoder without violating the channel's input constraint. The main feature of the scheme is a dynamic set of possible messages that is known both to the encoder and to the decoder at all times. Both parties will systematically reduce the size of the set of possible messages from $2^{nR}$ in the beginning of the transmission process to a single message that will then be announced as the correct message.

 \begin{figure}
  \centering
  \psfrag{A}[c][][1.2]{$l_j$}
  \psfrag{B}[c][][.8]{$\delta_j$}
  \psfrag{C}[c][][.8]{$0$}
  \psfrag{D}[c][][.8]{$1$}
  \psfrag{E}[c][][1.2]{$\tilde{l}_0$}
  \psfrag{F}[c][][.8]{$\overline{\delta}_0$}
  \includegraphics[scale=0.8]{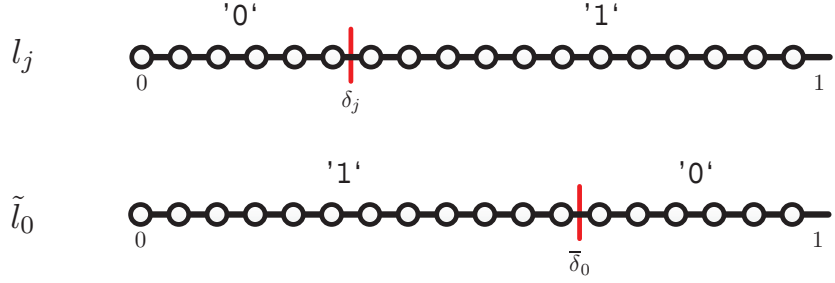}
  \caption{Labelings used in the coding scheme, with $j=1,\ldots,k$. Each subsection of $[0,1]$ is labeled with `$0$' or `$1$'.}\label{figures:Labeligns}
  \end{figure}

 Initially, the messages are mapped uniformly to message points in the unit interval by applying $m \mapsto \frac{m-1}{2^{nR}}$. As long as  transmission proceeds, the set of possible messages is represented by uniform points on the unit interval with proper scaling.

 Channel inputs are determined by $k+2$ \textit{labeling} functions, which map the unit interval into $\mathcal{X}$. Given a labelling, $l_j$, with a corresponding parameter $\delta_j$ the encoder assigns the label `$0$' to a subsection of $[0,1]$ of length $\delta_j$ and the label `$1$' to the rest of $[0,1]$. Fig. \ref{figures:Labeligns} depicts the various labelings. Define the following function:
 \begin{equation}\label{eq:X_function}
   X(m,L) = \begin{cases}
                  0, & (L=\tilde{l}_0 \text{ and } m>\overline{\delta}_0) \text{ or } (L=l_j \text{ and } m<\delta_j, \ j=0,\ldots,k)  \\
                  1, & \mbox{otherwise}.
                \end{cases}
 \end{equation}
 The channel input is $X_i = X(m,L_i)$, where $m$ is the correct message point and $L_i$ is the labeling being used at time $i$.

The labelling at each time is a function of all channel outpouts and can be computed recursively from the previous channel input and the previous labelling. Therefore, the instantaneous labelling is available both to the encoder and the decoder. Transition between the various labelings is controlled by a finite-state machine (FSM), which is illustrated in Fig. \ref{figures:coding_scheme_graph}. Define the following function:
 \begin{equation}\label{eq:Label_function}
   G(L,Y) = \begin{cases}
              l_0, &  (L=l_k) \text{ or }(Y=1) \text{ or } (Y=? \text{ and }L=\tilde{l}_0) \\
              \tilde{l}_0, &  Y=? \text{ and } L\neq\tilde{l}_0 \\
              l_{j+1}, &  Y=0
            \end{cases},
 \end{equation}
 and thus, $L_{i+1} = G(L_i,Y_i)$.

 \begin{figure}
  \centering
  \psfrag{A}[c][][.8]{$l_0$}
  \psfrag{B}[c][][.8]{$l_1$}
  \psfrag{C}[c][][.8]{$0$}
  \psfrag{D}[c][][.8]{$1$}
  \psfrag{E}[c][][.8]{$\tilde{l}_0$}
  \psfrag{F}[c][][.8]{$?$}
  \psfrag{G}[c][r][.8]{$1/?$}
  \psfrag{H}[c][][.8]{$l_2$}
  \psfrag{I}[c][][.8]{$l_{k-1}$}
  \psfrag{J}[c][][.8]{$l_k$}
  \includegraphics[scale=0.8]{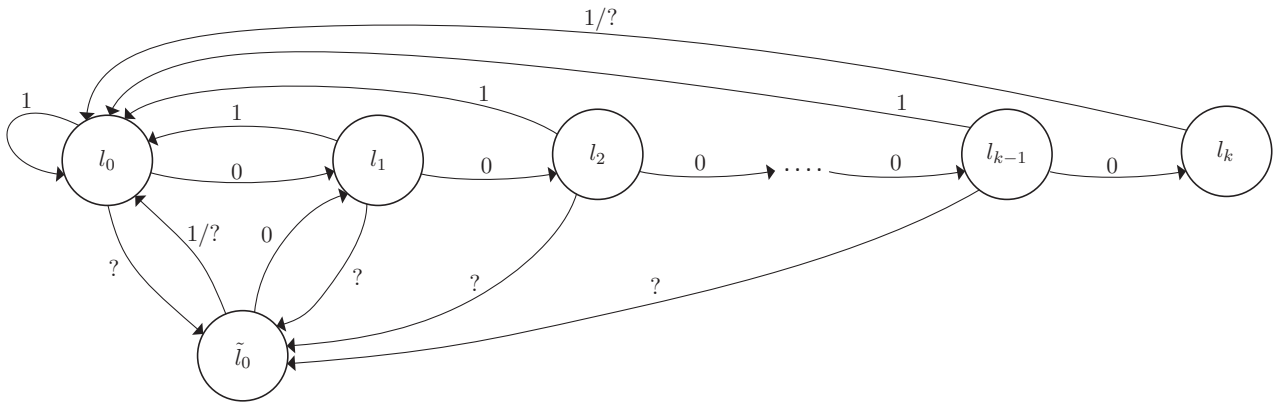}
  \caption{Finite state machine for the labelings transition. The nodes describe the instantaneous labelling that is used by the encoder. Edges correspond to channel outputs. Each node in the diagram corresponds to a labelling that can be calculated both at the encoder and at the decoder, since edges are a function of the outputs. }\label{figures:coding_scheme_graph}
  \end{figure}

A transmission at time $i$ is said to be successful if $Y_i\neq ?$. Due to the nature of the BEC, whenever $Y_i\neq ?$ the decoder can know with certainty the value of $X_i$. Denote by $\hat{M}_i^{(0)}$ and $\hat{M}_i^{(1)}$ the subsets of messages which are labeled at time $i$ with `$0$' and `$1$', respectively. Define $\hat{M}_0=\mathcal{M}$ and for $i \geq 1$:
\begin{equation}\label{eq:set_of_messages}
  \hat{M}_i = \begin{cases}
          \hat{M}_{i-1}, & Y_i=? \\
          \hat{M}_{i-1}^{(0)}, & Y_i=0 \\
          \hat{M}_{i-1}^{(1)}, & Y_i=1.
        \end{cases}
\end{equation}
Thus, a successful transmission reduces the size of the set of possible messages. Following a successful transmission, the remaining messages in the set of possible messages are uniformly mapped again to $[0,1)$. Fig. \ref{figures:successgul_bit} depicts a successful transmission and the subsequent reduction of the number of possible messages. The process continues until such a time that the set of possible messages contains only one message, at which point the decoder declares it to be $\hat{m}$.

 \begin{figure}
  \centering
  \psfrag{A}[c][][.8]{$\delta_0$}
  \psfrag{B}[c][][.8]{$\delta_1$}
  \psfrag{C}[c][][.8]{$0$}
  \psfrag{D}[c][][.8]{$1$}
  \psfrag{E}[c][][.8]{$i=1$}
  \psfrag{F}[c][][.8]{$Y_1=0$}
  \psfrag{G}[c][r][.8]{$i=2$}
  \psfrag{H}[c][r][.8]{$L=l_0$}
  \psfrag{I}[c][r][.8]{$L=l_1$}
  \psfrag{J}[c][r][.8]{$X_1=0$}
  \psfrag{K}[c][r][.8]{$X_2=1$}
  \includegraphics[scale=0.65]{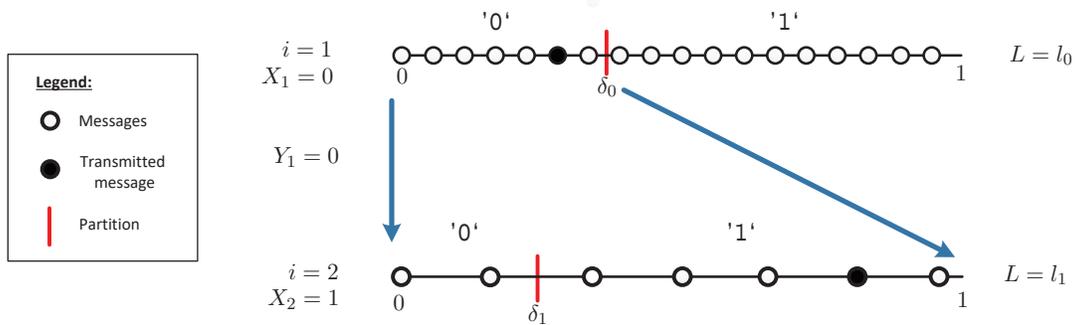}
  \caption{Example of a successful transmission. The black dot is the correct message point. At time instance $i=1$, the labeling is $L=l_0$ and the message point is labeled with `$0$' since it lies within $[0,\delta_0)$; thus, the encoder transmits $X_1=0$. Assume that $Y_1=0$. Consequently, the messages that were labeled with `$1$' are discarded, and the remaining messages are repositioned uniformly across $[0,1]$. These messages are $\hat{M}_1$. In accordance with the FSM in Fig \ref{figures:coding_scheme_graph}, the next label is $L=l_1$. At $i=2$, the message point is labeled with `$1$' since it now lies within $[\delta_1,1)$; thus, the encoder will transmit $X_2=1$.} \label{figures:successgul_bit}
  \end{figure}

Algorithm \ref{alg:Coding_Scheme} presents the coding scheme in pseudo code form. The functions $X(\cdot,\cdot)$ and $G(\cdot,\cdot)$ mentioned in the algorithm are defined in Eq. \eqref{eq:X_function} and Eq. \eqref{eq:Label_function}, respectively.

 \begin{algorithm}
\caption{Coding Scheme}\label{alg:Coding_Scheme}
\begin{algorithmic}
\State Inputs: $m$ - correct message
\State $\hat{\mathcal{M}} = \mathcal{M}$
\State $\text{Label} = l_0$
\While {$|\hat{M}|>1$}
    \State Transmit $X(\text{Label},m)$ \Comment{Encoder operation}
    \If {$Y=0$}
    \State $\hat{M}=\hat{M}^{(0)}$
\ElsIf {$Y=1$}
    \State $\hat{M}=\hat{M}^{(1)}$
\EndIf
\State $\text{Label} = G(\text{Label},Y)$
\EndWhile
\State $\hat{m} = \hat{\mathcal{M}}$ \Comment{Decoder operation}
\end{algorithmic}
\end{algorithm}

\subsection{Feasibility of the Proposed Scheme}
First, we show that the coding scheme satisfies the $(0,k)$-RLL constraint, that is, no message is mapped by the scheme into a sequence with more than $k$ consecutive `$0$'s. The following lemma shows that the constraint is satisfied when restricting the scheme parameters $\delta_j$.

 \begin{lemma}\label{lemma:feasability}
   If $\delta_j\leq\frac{1}{2}$ for $j=0,\ldots,k-1$, then any channel input sequence generated by the proposed coding scheme satisfies the $(0,k)$-RLL constraint.
 \end{lemma}
 \begin{proof}
 We show that if $\delta_j\leq\frac{1}{2}$ for $j=0,\ldots,k-1$, then no message is labeled `$0$' more than $k$ times in a row. From Eqs. \eqref{eq:X_function} and \eqref{eq:Label_function}, the channel input, $X_i$ is a function of $m$ and $y^{i-1}_{i-k-1}$. Therefore, we divide the proof into three disjoint cases based on the $k$ last outputs $y^{i-1}_{i-k-1}$. For each case, we show that the subsequent sequence of channel inputs cannot contain more than $k$ consecutive `$0$'s. Assume that transmission begins at the node associated with labeling $l_0$.
 \begin{enumerate}
   \item Any output sequence (of length $k$) that contains a `$1$' cannot cause a violation.
   \item An output sequence of $k$ consecutive `$0$'s ends at $l_k$ (Fig. \ref{figures:coding_scheme_graph}). Thus, the next channel input is $X=1$ (Fig \ref{figures:Labeligns}).
   \item Lastly, consider a sequence of $k$ outputs that contains both `$0$'s and erasures. Assume the first erasure occurred at time instance $i$, meaning that the erasure took place while the encoder was using labeling $l_i$. This means that all messages between $0$ and $\delta_i$ in $[0,1)$ were labeled with `$0$' $i+1$ times in a row. The next labeling that will be used is $\tilde{l}_0$. In this labeling, all messages between $0$ and $\overline{\delta}_0$ are labeled with `$1$'. Since $\delta_0 \leq \frac{1}{2}$, we have that $\overline{\delta}_0 \geq \frac{1}{2}$, so all messages that were labeled `$0$' $i+1$ times in a row will be labeled `$1$' in $\tilde{l}_0$. This analysis holds for any $i=0,\ldots,k-1$.
 \end{enumerate}
 In summary, setting $\delta_i \leq \frac{1}{2}$, $i=0,\ldots,k-1$ ensures that the scheme does not violate the $(0,k)$-RLL constraint.
\end{proof}

 \subsection{Rate Analysis}
 The achieved rate $R$ is measured by $\frac{\text{expected number of information bits}}{\text{expected number of channel uses}}$. Define $Q$ as the number of information bits gained in a single channel use, i.e., the quotient of the size of the set of possible messages before and after the transmission. Additionally, Let $L$ be the random variable which corresponds to the labeling and takes values in $\mathcal{L} = \{\tilde{l}_0,l_0,\ldots,l_k\}$.

 In the following lemma we calculate the expectation of $Q \mid L=l$.

 \begin{lemma}\label{lemma:expectation_Q|L=l}
   For all $l\in\mathcal{L}$, we have that $\mathbb{E}[Q \mid L=l] = \overline{\varepsilon}H_2(\delta_l)$, where $\delta_l$ is the $\delta$ relevant to the labeling $l$.
 \end{lemma}
 \begin{proof}
 Consider:
   \begin{align}\label{eq:expectation_Q|L=l}
     \mathbb{E}[Q \mid L=l] &= \varepsilon\mathbb{E}[Q \mid L=l,\theta= \text{\xmark}]+\overline{\varepsilon}\mathbb{E}[Q \mid L=l,\theta= \checkmark] \nonumber \\
      &\stackrel{(a)}{=} \overline{\varepsilon}\mathbb{E}[Q \mid L=l,\theta= \checkmark] ,
   \end{align}
   where (a) holds since if $\theta=\text{\xmark}$, then the transmitted symbol is erased by the channel and the set of possible messages is unchanged.

   In the proposed coding scheme, labeling $l_j$ assigns a portion of $[0,1]$ of size $\delta_j$ to the label `$0$' and the rest to label `$1$' for all $j=0,\ldots,k$. The labeling $\tilde{l}_0$ also assigns $\delta_0$ of the unit interval to `$0$'. If the labeling $l_j$ is employed, then the channel input is distributed according to $Ber(\overline{\delta}_j)$\footnote{For labelings $\tilde{l}_0,l_0,l_1,\ldots,l_{k-1}$, the encoder transmits $X=1$ if the correct message falls within a sub-interval of $[0,1)$ that has length $\overline{\delta}_0,\overline{\delta}_0,\overline{\delta}_1,\ldots,\overline{\delta}_{k-1},$ respectively. Note that the messages are discrete points in $[0,1)$ and it is possible for the partition to occur between two messages. This implies that the transmitted bit is distributed $Ber(\overline{\delta}_i+e_i)$, where $e_i$ is a correction factor. From the continuity of the entropy function, the contribution of this correction factor can be bounded with arbitrary constant by taking the block length $n$ be large enough. The precise details are omitted and follow parallel argument to \cite[Appendix C]{Sabag_BEC}.}

   Assume that $|\hat{\mathcal{M}}|=a$. If the successfully received bit was `$1$', then the new set of possible messages has size $\overline{\delta}_la$, and if it was `$0$', then the new set of possible messages has size $\delta_la$. The expected number of bits required to describe the new set of possible messages is $\overline{\delta}_l\log_2(\overline{\delta}_la)+\delta_l\log_2(\delta_la)=\log_2(a)-H_2(\delta_l)$. Thus, given that $L=l$, following a successful transmission the decoder gains $H_2(\delta_l)$ bits of information. Substituting into \eqref{eq:expectation_Q|L=l} we get:
   \begin{equation*}
     \mathbb{E}[Q \mid L=l] = \overline{\varepsilon}H_2(\delta_l).
   \end{equation*}
 \end{proof}

  The next lemma calculates the rate achieved by the proposed coding scheme.
  \begin{lemma}\label{lemma:rate}
    For any $\varepsilon\in[0,1]$, $k\geq1$ and $0\leq\delta_0,\ldots,\delta_{k-1}\leq\frac{1}{2}$, the proposed coding scheme achieves the following rate $R$:
    \begin{equation}\label{eq:rate_in_lemma}
      R = \frac{\overline{\varepsilon}H_2(\delta_0)+\displaystyle\sum_{i=1}^{k-1}\left(\overline{\varepsilon}^{i+1}H_2(\delta_i)\prod_{m=0}^{i-1}\delta_m\right)} {1+\displaystyle\sum_{i=0}^{k-1}\left(\overline{\varepsilon}^{i+1}\prod_{m=0}^{i}\delta_m\right)}.
    \end{equation}
  \end{lemma}

  \begin{proof}
  Consider the averaged gain of information divided by the amount of time:
  \begin{align*}
    R &= \lim_{n\to\infty}\frac{1}{n}\sum_{j=1}^{n}\mathbb{E}[Q_j] \\
     &\stackrel{(a)}{=}  \lim_{n\to\infty}\frac{1}{n}\sum_{j=1}^{n}\sum_{l\in\mathcal{L}}P(L_j=l)\mathbb{E}[Q_j \mid L_j=l] \\
     &\stackrel{(b)}{=} \sum_{l\in\mathcal{L}}\overline{\varepsilon}H_2(\delta_l)\lim_{n\to\infty}\frac{1}{n}\sum_{j=1}^{n}P(L_j=l) \\
     &\stackrel{(c)}{=} \sum_{l\in\mathcal{L}}\overline{\varepsilon}H_2(\delta_l)\pi(l) \\
     &= \overline{\varepsilon}H_2(\delta_0)\left( \pi(\tilde{l}_0)+\pi(l_0) \right)+\sum_{i=1}^{k-1}\overline{\varepsilon}H_2(\delta_i)\pi(l_i) \\
     &\stackrel{(d)}{=} \frac{\overline{\varepsilon}H_2(\delta_0)+\displaystyle\sum_{i=1}^{k-1}\left(\overline{\varepsilon}^{i+1}H_2(\delta_i)\prod_{m=0}^{i-1}\delta_m\right)}{1+\displaystyle\sum_{i=0}^{k-1}\left(\overline{\varepsilon}^{i+1}
   \prod_{m=0}^{i}\delta_m\right)},
  \end{align*}
  where
  \begin{itemize}
    \item [(a)] Follows from the law of total expectation.
    \item [(b)] Follows from Lemma \ref{lemma:expectation_Q|L=l} and exchanging the finite sums' order.
    \item [(c)] Follows from the definition of stationary probability. $\pi(l_i)$ is the stationary probability of labeling $l_i$. There exists a stationary probability because the random process $\{L_j\}$ is a positive recurrent, irreducible and aperiodic Markov chain, as can be seen from Fig. \ref{figures:coding_scheme_graph}.
        % Regarding the continuity of the stationary distribution w.r.t the deltas: Theorem 12.13 in this book https://www.scribd.com/doc/148731582/Stokey-Lucas-Recursive-Methods-in-Economic-Dynamics-1989 (page 394 in the pdf) seems to say that it's okay since our state space is compact (it's actually finite for us...)
    \item [(d)] Follows from calculation of the stationary probability of the Markov chain described in Fig. \ref{figures:coding_scheme_graph} and is parameterized with $\delta_j$, $j=0,\ldots,k-1$. Calculating the conditional probability of each edge is simple, using the law of total probability. For example, the conditional distribution of the edge beginning node $l_0$ and culminating in node $l_1$ is $\overline{\varepsilon}\delta_0$.
  \end{itemize}
   From Lemma \ref{lemma:feasability}, we conclude that $\max_{0\leq\delta_0,\ldots,\delta_{k-1}\leq\frac{1}{2}} R_\varepsilon(\delta_0,\ldots,\delta_{k-1})$ is an achievable rate.
%   Therefore:
%   \begin{equation}\label{eq:lower_bound}
%   C^{\mathrm{fb}}_{(0,k)}(\varepsilon)\geq max_{0\leq\delta_0,\ldots,\delta_{k-1}\leq\frac{1}{2}} R_\varepsilon(\delta_0,\ldots,\delta_{k-1}).
% \end{equation}
 \end{proof}

\section{Non-Causal Upper Bound}\label{sec:upper_bound}
In this section, we present an upper bound of the non-causal capacity for the $(0,k)$-RLL input-constrained BEC (given in Eq. \eqref{eq:inequalities}). We begin with an observation: it is sufficient to look at a smaller family of codes, called  \emph{restricted codes}. We then proceed to calculate an upper bound on the achievable rate of such codes, using standard converse arguments, as well as the method of types and Markov theory results.

A code is said to be \textit{restricted} if
\begin{align}\label{eq:def_restricted}
  &g_i(m,\theta^{i-1},\theta_i=\text{\xmark})=1, \ \forall m,\theta^{i-1}, i=1,\ldots,n.
\end{align}
Condition \eqref{eq:def_restricted} states that if an erasure is about to occur, the encoder transmits $X=1$. The following lemma formalizes the fact that restricted codes can achieve the capacity.
\begin{lemma}\label{lemma:preliminary_lemma}
   For the $(0,k)$-RLL constrained non-causal BEC, if a rate $R$ is achievable, then $R$ can be achieved using a sequence of restricted codes.
   \begin{proof}
     Assume the rate $R$ is achieved using a sequence of codes: $C_n$. Define for each $n$ a new code $C'_n$ that is exactly the same as $C_n$ except that in $C'_n$, whenever $\theta_i=\text{\xmark}$ the encoder transmits $x_i=1$.

     The code $C'_n$ does not violate the input constraint since the original $C_n$ did not violate the constraint and transmitting $x_i=1$ is always permitted by the $(0,k)$-RLL input constraint. In addition, the channel outputs remain the same whether the code is $C_n$ or $C'_n$. This means that $P_e^{(n)}(C'_n)=P_e^{(n)}(C_n)$, and so the rate $R$ is also achieved by the sequence of $C'_n$.
   \end{proof}
 \end{lemma}

 %Note that since the inherent behavior of $\tilde{Q}_{Y^n}$ is that of a time-invariant Markov chain, its entropy rate can be made arbitrarily close to the entropy rate of a stationary Markov chain \textcolor[rgb]{1.00,0.00,0.00}{[Cover Book with relevant exercise].}

\subsection{Upper Bound Calculation}
The following technical lemma is needed for the converse proof:
 \begin{lemma}\label{lemma:n-tuple_entropy}
 For any n-tuple constrained distribution $\tilde{P}_{Y^n}(y^n)=\mathbbm{1}_{\{y_1=1\}}\prod_{i=2}^{n}\tilde{P}^{(i)}_{Y_i \mid Y_{i-1}}(y_i \mid y_{i-1})$, where $\tilde{P}^{(i)}_{Y_i \mid Y_{i-1}}(y_i=0 \mid y_{i-1}=0)=0 \quad \forall i=2,\ldots,n$, there exists a time invariant constrained Markov distribution $\tilde{Q}_{Y^n}(y^n)=\prod_{i=1}^{n}\tilde{P}_{Y_i \mid Y_{i-1}}(y_i \mid y_{i-1})$ such that:
    \begin{equation}\label{eq:n-tuple_entropy}
      H_{\tilde{P}}(Y^n)\leq H_{\tilde{Q}}(Y^n)+\zeta_n,
    \end{equation}
    where $\lim_{n\to\infty}\zeta_n=0$.
 \end{lemma}
 The proof is available in Appendix \ref{sec:Lemma}. This result can readily be generalized to any $(d,k)$-RLL constraint imposed on the n-tuple distribution.

In the constraint graph, Fig. \ref{figures:(0,k)_RLL}, a node $S_i$ can be calculated from an any-length tuple $X^{i-1}$ by walking along the edges labelled with $X^{i-1}$, since we assume that the initial state is $s_0=0$. The notation $\tilde{P}$ will be used in various forms for distributions on $\mathcal{Y}$ to signify that the probability for $?$ is $\varepsilon$ and that the probability for a constrained word is $0$.
\begin{proof}[Proof of the upper bound]
Let $R$ be an achievable rate using a restricted code, and consider the following chain of inequalities:
 \begin{align}\label{1}
   nR &= H(M) \nonumber \\
    &\stackrel{(a)}{\leq} I(Y^n;M) + \varepsilon_n \nonumber \\
    &\stackrel{(b)}{=} H(Y^n) - \sum_{i=1}^{n}H(Y_i \mid M,Y^{i-1}) +\varepsilon_n \nonumber \\
    &\stackrel{(c)}{\leq} H(Y^n) - nH_2(\varepsilon) +\varepsilon_n \nonumber \\
    &\stackrel{(d)}{\leq} \sum_{i=1}^{n}H(Y_i \mid Y_{i-k}^{i-1}) - nH_2(\varepsilon) +\varepsilon_n \nonumber \\
    &\stackrel{(e)}{\leq} \max_{\{\tilde{P}_i(y_i \mid y^{i-1}) \}_{i=1}^n} \sum_{i=1}^{n}H(Y_i \mid Y_{i-k}^{i-1}) - nH_2(\varepsilon) +\varepsilon_n \nonumber \\
    &\stackrel{(f)}{=} \max_{\{\tilde{P}_i(y_i \mid y_{i-k}^{i-1}) \}_{i=1}^n} \sum_{i=1}^{n}H(Y_i \mid Y_{i-k}^{i-1}) - nH_2(\varepsilon) +\varepsilon_n \nonumber \\
    &\stackrel{(g)}{\leq} \max_{\{\tilde{P}(y_i \mid y_{i-k}^{i-1}) \}_{i=1}^n} \sum_{i=1}^{n}H(Y_i \mid Y_{i-k}^{i-1}) - nH_2(\varepsilon) +\varepsilon'_n \nonumber \\
    %&\stackrel{(h)}= \max_{\{\tilde{P}(y_i \mid y_{i-k}^{i-1},x_{i-k}^{i-1})\}_{i=1}^n}\sum_{i=1}^{n}H(Y_i \mid Y_{i-k}^{i-1},X_{i-k}^{i-1}) - nH_2(\varepsilon) +\varepsilon'_n \nonumber \\
    &\stackrel{(h)}= \max_{\{\tilde{P}(y_i \mid y_{i-k}^{i-1},s_{i-1})\}_{i=1}^n}\sum_{i=1}^{n}H(Y_i \mid Y_{i-k}^{i-1},S_{i-1}) - nH_2(\varepsilon) +\varepsilon'_n \nonumber \\
    &\stackrel{(i)}{\leq} \max_{\{\tilde{P}(y_i \mid y_{i-k}^{i-1},s_{i-1})\}_{i=1}^n}\sum_{i=1}^{n}H(Y_i \mid S_{i-1}) - nH_2(\varepsilon) +\varepsilon'_n \nonumber \\
    &\stackrel{(j)}{=} \max_{\{\tilde{P}(y_i \mid s_{i-1})\}_{i=1}^n}\sum_{i=1}^{n}H(Y_i \mid S_{i-1}) - nH_2(\varepsilon) +\varepsilon'_n \nonumber \\
    &\stackrel{(k)}{=} \max_{0\leq\delta_0,\ldots,\delta_{k-1}\leq1}\sum_{i=1}^{n}\sum_{j=0}^{k-1} \Pr(S_{i-1}=j)H_3(\overline{\varepsilon}\delta_j,\varepsilon,\overline{\varepsilon}\overline{\delta}_j) - nH_2(\varepsilon) +\varepsilon'_n \nonumber \\
    &\stackrel{(l)}{=} \max_{0\leq\delta_0,\ldots,\delta_{k-1}\leq1}\sum_{i=1}^{n}\sum_{j=0}^{k-1}\Pr(S_{i-1}=j)\left[ H_2(\varepsilon)+\overline{\varepsilon}H_2(\delta_j) \right] - nH_2(\varepsilon) +\varepsilon'_n \nonumber \\
    &\stackrel{} = \max_{0\leq\delta_0,\ldots,\delta_{k-1}\leq1}\overline{\varepsilon}\sum_{j=0}^{k-1}H_2(\delta_j)\sum_{i=1}^{n}\Pr(S_{i-1}=j) +\varepsilon'_n.
 \end{align}
 where
 \begin{enumerate}
   \item[(a)] Follows from Fano's inequality.
   \item[(b)] Follows from the chain rule.
   \item[(c)] Follows from the fact that conditioning reduces entropy, so: $H(Y_i \mid M,Y^{i-1}) \geq H(Y_i \mid X_i,M,Y^{i-1}) = H_2(\varepsilon)$.
   \item[(d)] Follows from the fact that conditioning reduces entropy.
   \item[(e)] The maximization domain is the set of all n-tuple distributions $\tilde{P}(y^n)$ which induce $\tilde{P}(y_i=?  \mid  y^{i-1})=\varepsilon$ and $\tilde{P}(y_i=0  \mid  y_{i-k}^{i-1}=0^k)=0$, for all $i=1,\ldots,n$ and $i=k+1,\ldots,n$, respectively.
   \item[(f)] We want to show that it is possible to maximize over a smaller domain and maintain an equality. It Suffices to prove by induction that if we have two distributions $\{\tilde{P}^{(1)}(y_i \mid y^{i-1})\}_{i\geq1}$ and $\{\tilde{P}^{(2)}(y_i \mid y^{i-1})\}_{i\geq1}$, which induce the same marginal distributions $\{\tilde{P}(y_i \mid y_{i-k}^{i-1})\}_{i\geq1}$, then $\{\tilde{P}^{(1)}(y_{i-k}^i)\}_{i\geq1}$ and $\{\tilde{P}^{(2)}(y_{i-k}^i)\}_{i\geq1}$ coincide. For $i=1$ the proof is trivial. Assume by induction that $\tilde{P}^{(1)}(y_{i-1-k}^{i-1})=\tilde{P}^{(2)}(y_{i-1-k}^{i-1})$ and we need to prove that $\tilde{P}^{(1)}(y_{i-k}^i)=\tilde{P}^{(2)}(y_{i-k}^i)$. Indeed we have:
       \begin{equation*}
         \tilde{P}^{(1)}(y_{i-k}^i)=\tilde{P}^{(1)}(y_{i-k}^{i-1})\tilde{P}(y_i \mid y_{i-k}^{i-1})=\tilde{P}^{(2)}(y_{i-k}^{i-1})\tilde{P}(y_i \mid y_{i-k}^{i-1})=\tilde{P}^{(2)}(y_{i-k}^i),
       \end{equation*}
       since $\tilde{P}(y_i \mid y_{i-k}^{i-1})$ is the same for both distributions by assumption, and the induction assumption tells us that $\tilde{P}^{(1)}(y_{i-1-k}^{i-1})=\tilde{P}^{(2)}(y_{i-1-k}^{i-1})$, and thus we have $\tilde{P}^{(1)}(y_{i-k}^{i-1})=\tilde{P}^{(2)}(y_{i-k}^{i-1})$ as well. Additionally, it can easily be shown that $\tilde{P}(y_i=?  \mid  y_{i-k}^{i-1})=\varepsilon$ and $\tilde{P}(y_i=0  \mid  y_{i-k}^{i-1}=0^k)=0$, for all $i=1,\ldots,n$ and $i=k+1,\ldots,n$, respectively.
   \item[(g)] Follows from Lemma \ref{lemma:n-tuple_entropy}. Notice that the distributions in the maximization domain are now time-invariant.
   \item[(h)] Follows from the fact that $S_{i-1}$ is a function of $Y_{i-k}^{i-1}$: since the code is restricted, $X_{i-k}^{i-1}$ is a function of $Y_{i-k}^{i-1}$ and, by its definition, $S_{i-1}$ is a function of $X_{i-k}^{i-1}$.
   \item[(i)] Follows from the fact that conditioning reduces entropy.
   \item[(j)] Similarly to step (f), it suffices to prove by induction that if we have two distributions $\{\tilde{P}^{(1)}(y_i \mid y_{i-k}^{i-1},s_{i-1})\}_{i=1}^n$ and $\{\tilde{P}^{(2)}(y_i \mid y_{i-k}^{i-1},s_{i-1})\}_{i=1}^n$ which induce the same marginal distributions $\{\tilde{P}(y_i \mid s_{i-1})\}_{i=1}^n$, then $\{\tilde{P}^{(1)}(y_{i-k}^i,s_{i-1}^{i})\}_{i=1}^n$ and $\{\tilde{P}^{(2)}(y_{i-k}^i,s_{i-1}^{i})\}_{i=1}^n$ coincide. Recall that since the code is restricted, $s_i$ is a function of $(s_{i-1},y_i)$. denote this function by $s_i=h(s_{i-1},y_i)$. For $i=1$ we have:
       \begin{align*}
         \tilde{P}^{(1)}(y_1,s_0^1) &= \mathbbm{1}_{\{s_0=0\}}\tilde{P}(y_1 \mid s_0)\tilde{P}^{(1)}(s_1 \mid y_1,s_0) \\
         &= \mathbbm{1}_{\{s_0=0\}}\tilde{P}(y_1 \mid s_0)\mathbbm{1}_{s_1=h(s_0,y_1)}=\tilde{P}^{(2)}(y_1,s_0^1).
       \end{align*}
        Now, assume by induction that $\tilde{P}^{(1)}(y_{i-1-k}^{i-1},s_{i-2}^{i-1})$ and $\tilde{P}^{(2)}(y_{i-1-k}^{i-1},s_{i-2}^{i-1})$ and we need to prove that $\tilde{P}^{(1)}(y_{i-k}^i,s_{i-1}^{i})=\tilde{P}^{(2)}(y_{i-k}^i,s_{i-1}^{i})$:
       \begin{align*}
         \tilde{P}^{(1)}(y_{i-k}^i,s_{i-1}^{i}) &= \tilde{P}^{(1)}(y_{i-k}^{i-1},s_{i-1}) \tilde{P}^{(1)}(y_i \mid y_{i-k}^{i-1},s_{i-1}) \tilde{P}^{(1)}(s_i \mid y_{i-k}^i,s_{i-1}) \\
          &\stackrel{(1)}{=} \tilde{P}^{(2)}(y_{i-k}^{i-1},s_{i-1})P(y_i \mid s_{i-1})\mathbbm{1}_{\{s_i=h(s_{i-1},y_i)\}} \\
          &= \tilde{P}^{(2)}(y_{i-k}^i,s_{i-1}^{i}),
       \end{align*}
       where $(1)$ follows from the induction assumption, the Markov chain $y_i-s_{i-1}-y_{i-k}^{i-1}$ and the notation defined above.
   \item[(k)] Follows by defining a conditional distribution, $\delta_j\triangleq p(X=0 \mid S=j,\theta=\checkmark)$.
   \item[(l)] Follows from a simple identity.
 \end{enumerate}

 For each instance of the tuple $(\delta_{0},\dots,\delta_{k-1})$, the random process $\{S_i\}_{i=1}^n$ is first-order Markov. Additionally, for all tuples, there is a single closed communicating class for this process, so there exists a stationary distribution and the value of $\sum_{i=1}^{n}  \Pr(S_i=j)$ can be made
 arbitrarily close to $n\pi_S(j)$, where $\pi_S(j)$ denotes the stationary distribution that is induced by the Markov chain in Fig. \ref{figures:(0,k)_RLL}. Using the transitions matrix of the Markov process $\{S_i\}_{i=1}^n$:
  \begin{center}
   \begin{tabular}{c|ccccccc}

     \backslashbox{$S_{i-1}$}{$S_i$} & $S=0$ & $S=1$ & $S=2$ & $S=3$ & \dots & $S={k-1}$ & $S=k$ \\
     \hline
     $S=0$ & $\varepsilon+\bar{\varepsilon}\delta_0$ & $\bar{\varepsilon}\bar{\delta}_0$ & 0 & 0 & \dots & 0 & 0 \\
     $S=1$ & $\varepsilon+\bar{\varepsilon}\delta_1$ & 0 & $\bar{\varepsilon}\bar{\delta}_1$ & 0 & \dots & 0 & 0 \\
     $S=2$ & $\varepsilon+\bar{\varepsilon}\delta_2$ & 0 & 0 & $\bar{\varepsilon}\bar{\delta}_2$ & \dots & 0 & 0 \\
     $\vdots$ & $\vdots$ & $\vdots$ & $\vdots$ & $\vdots$ & $\ddots$ & $\vdots$ & $\vdots$ \\
     $S={k-1}$ & $\varepsilon+\bar{\varepsilon}\delta_{k-1}$ & 0 & 0 & 0 & \dots & 0 & $\bar{\varepsilon}\bar{\delta}_{k-1}$ \\
     $S=k$ & 1 & 0 & 0 & 0 & \dots & 0 & 0 \\

   \end{tabular}
 \end{center}
 we can show that:
\begin{equation*}
  \pi_S(j) = \frac{\overline{\varepsilon}^{j}\prod_{m=0}^{j-1}\delta_m}{1+\sum_{j=0}^{k-1}\left(\overline{\varepsilon}^{j+1}\prod_{m=0}^{j}\delta_m\right)}, \ j=0,\ldots,k-1 ,
\end{equation*}
where $\prod_{m=0}^{-1}\delta_m\triangleq 1$. Therefore, we have that
 \begin{align*}
   R &\leq (1-\varepsilon) \max_{\{p_i(x \mid s,\theta=\checkmark)\}}  \sum_{j=0}^{k-1} H_2(\delta_j) [\pi_S(j)+\varepsilon''_n] + \frac{\epsilon'_n}{n},
 \end{align*}
where $\varepsilon''_n$ is the correcting factor from the stationary distribution and satisfies $\varepsilon''_n \to 0$. By taking the limit $n\to\infty$, and substituting the stationary distribution, we conclude that an achievable rate is upper bounded by

 \begin{equation}\label{eq:upper_bound}
    C^{\mathrm{nc}}_{(0,k)} \leq \max_{0\leq\delta_0,\ldots,\delta_{k-1}\leq 1} R_\varepsilon(\delta_0,\ldots,\delta_{k-1}).
 \end{equation}
\end{proof}

\section{Does $C^{\mathrm{fb}}=C^{\mathrm{nc}}$ for any input constraint?}\label{sec:(2,infty)-RLL}
In previous sections it was shown that $C^{\mathrm{fb}}_{(0,k)}(\varepsilon) = C^{\mathrm{nc}}_{(0,k)}(\varepsilon)$. This section concerns the ensuing question: is it true that non-causal knowledge of the upcoming erasure does not increase the feedback capacity for any input constraint?

It turns out that the non-causal capacity of the $(d,\infty)$-RLL case can be easily solved using the same arguments as in previous sections. Therefore, we investigated this family with a hope to prove its feedback capacity as well. For $d=1$, it has been proven in \cite{Sabag_BEC} that $C^{\mathrm{fb}}_{(1,\infty)}(\varepsilon)=C^{\mathrm{cb}}_{(1,\infty)}(\varepsilon)$. This result coincides with Theorem \ref{thm:fb_capacity=nc_capacity} since the $(1,\infty)$ constraint is equivalent to the $(0,1)$ constraint by swapping `$1$'s and `$0$'s. However, for d=2, we are able to show that  $C_{(2,\infty)}^{\mathrm{fb}}(\varepsilon)<C_{(2,\infty)}^{\mathrm{nc}}(\varepsilon)$. Thus, the answer to the aforementioned question is no.

The first result is the non-causal capacity of the BEC with a $(d,\infty)$-RLL input constraint, for any $d\geq 1$.
\begin{lemma}\label{lemma:(d,inf)_RLL}
  For any $d\in\mathbb{N}$, the non-causal capacity of the $(d,\infty)$-RLL input constrained BEC is given by:
  \begin{equation}\label{eq:(d_inf)_nc_capacity}
    C^{\mathrm{nc}}_{(d,\infty)}(\varepsilon) = \max_{0\leq\delta\leq\frac{1}{2}}\frac{H_2(\delta)}{\frac{1}{1-\varepsilon}+d\delta}.
  \end{equation}
\end{lemma}
\begin{proof}
  The upper bound of $C^{\mathrm{nc}}_{(d,\infty)}(\varepsilon)$ is derived following the same steps presented in Section \ref{sec:upper_bound}. In this case, a restricted encoder that transmits $X=0$ whenever $\theta=$ \xmark. The rest of the proof mirrors that of Section \ref{sec:upper_bound}, and we are able to show that $C^{\mathrm{nc}}_{(d,\infty)}(\varepsilon) \leq \max_{0\leq\delta\leq\frac{1}{2}}\frac{H_2(\delta)}{\frac{1}{1-\varepsilon}+d\delta}$. This expression is also a lower bound. It is achieved by applying a restricted encoder which transmits $X\sim Ber(\delta)$ if an erasure does not occur. The expected number of information bits gained in a successful transmission is $H_2(\delta)$ and the expected number of channel uses to transmit successfully is $\frac{1}{1-\varepsilon}$ , plus another $d$ channel uses if the transmitted bit is a `$1$'.
\end{proof}

\begin{figure}
    \centering
    \psfrag{A}[c][][1]{$S=0$}
    \psfrag{B}[c][][1]{$S=1$}
    \psfrag{c}[c][][1]{$S=2$}
    \psfrag{E}[c][][1]{$0$}
    \psfrag{F}[c][][1]{$1$}
    \includegraphics[scale=1]{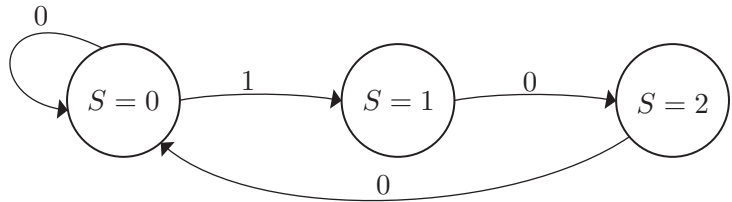}
    \caption{State diagram describing all sequences that can be generated while satisfying the $(2,\infty)$-RLL constraint: every `$1$' is followed by at least two `$0$'s.}\label{figures:(2,inf)_S_graph}
\end{figure}

Next we prove that $C^{\mathrm{fb}}_{(2,\infty)}(\varepsilon)$ is upper bounded by an expression which is strictly smaller than the RHS of Eq. \eqref{eq:(d_inf)_nc_capacity} for $d=2$. To discuss an upper bound for $C^{\mathrm{fb}}_{(2,\infty)}(\varepsilon)$, we must first introduce the concepts of the S-graph and the Q-graph. Fig. \ref{figures:(2,inf)_S_graph} contains an S-graph, which is simply a graphical representation of the $(2,\infty)$-RLL constraint. A Q-graph is an irreducible directed graph in which each node has $|\mathcal{Y}|$ distinct outgoing edges. The upper bound is derived using the method introduced in \cite{Q-graph}. This method involves a combined representation of both the S-graph and the Q-graph in a coupled $(S,Q)$-graph, which has a stationary distribution denoted $\pi(s,q)$. The main result in \cite{Q-graph} states the following:
\begin{theorem}[Theorem 2, \cite{Q-graph}]\label{thm:Oron_Q_graph}
  For every $Q$-graph, the feedback capacity is bounded by
    \begin{equation*}
        C^{\mathrm{fb}}\leq \sup_{p(x \mid s,q)}I(X;Y \mid Q),
    \end{equation*}
where $S$ represents the input constraint state. The joint distribution is $\pi(s,q)p(x \mid s,q)p(y \mid x,s)$.
\end{theorem}
We apply Theorem \ref{thm:Oron_Q_graph} with the Q-graph in Fig. \ref{figures:(2,inf)_Q_graph}. This graph was estimated from numerical evaluations of the associated DP problem.
\begin{figure}
    \centering
    \psfrag{A}[c][][1]{$q_1$}
    \psfrag{B}[c][][1]{$q_2$}
    \psfrag{C}[c][][1]{$q_3$}
    \psfrag{D}[c][][1]{$q_4$}
    \psfrag{E}[c][][1]{$q_5$}
    \psfrag{F}[c][][1]{$0$}
    \psfrag{G}[c][][1]{?}
    \psfrag{H}[c][][1]{$1$}
    \psfrag{I}[c][][1]{$0/?$}
    \psfrag{J}[c][][1]{$0/?/1$}
    \includegraphics[scale=1]{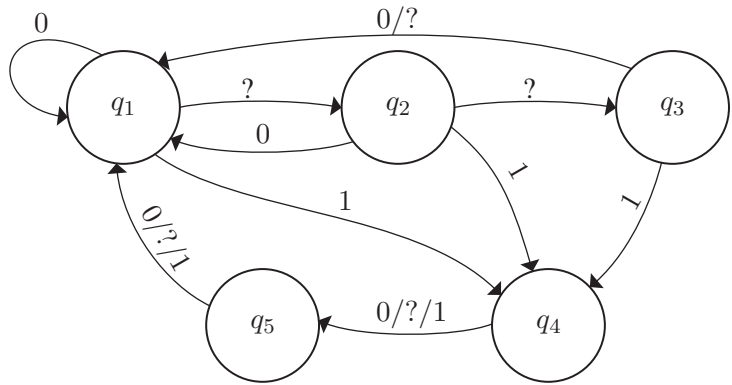}
    \caption{Q-graph for the $(2,\infty)$-RLL BEC}\label{figures:(2,inf)_Q_graph}
\end{figure}
Calculating $\sup_{p(x \mid s,q)}I(X;Y \mid Q)$ we get:
\begin{equation}\label{eq:(2,inf)_upper_bound}
  C^{\mathrm{fb}}_{(2,\infty)}(\varepsilon) \leq \max_{\substack{0\leq\delta_0,\delta_1,\delta_2\leq1 \\ \delta_0+\delta_1+\delta_2\leq1}}\frac{\overline{\varepsilon}\left(H_2(\delta_0)+\varepsilon H_2(\delta_1)+\varepsilon^2H_2(\delta_2)\right)}{1+\varepsilon+\varepsilon^2+2\overline{\varepsilon}(\delta_0+\varepsilon\delta_1+\varepsilon^2\delta_2)}.
\end{equation}

Fig. \ref{figures:nc_fb_comparison} contains graphs of the non-causal capacity and the feedback upper bound for $0\leq\varepsilon\leq1$.
\begin{figure}[h]
    \centering
    \psfrag{A}[][][0.8]{ }
    \psfrag{B}[][][0.8]{Erasure probability $\varepsilon$}
    \psfrag{C}[][][0.8]{Capacity and upper bound}
    \includegraphics[scale=0.5]{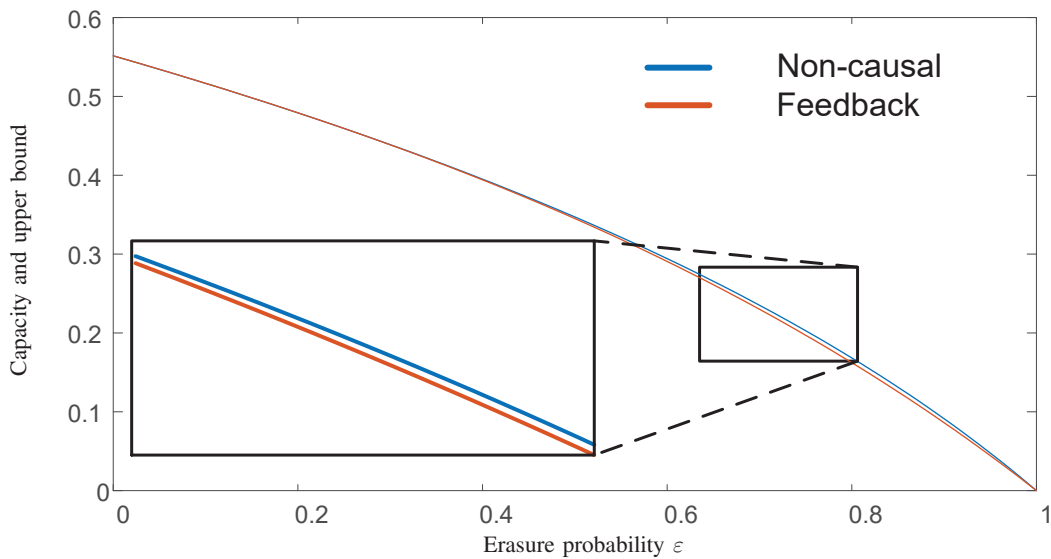}
    \caption{Non-causal capacity and feedback upper bound for the $(2,\infty)$-RLL input constrained BEC, as a function of $\varepsilon$. The non-causal capacity is greater than the upper bound of the feedback capacity. Note that $C^{\mathrm{nc}}_{(2,\infty)}(0)=C^{\mathrm{nc}}_{(2,\infty)}(0) \sim 0.551$, which is the $(2,\infty)$-RLL constraint capacity.} \label{figures:nc_fb_comparison}
\end{figure}
It is clear that the non-causal capacity is strictly greater than the feedback upper bound in the case of the $(2,\infty)$-RLL input constrained BEC. The following lemma states the strong inequality for a specific $\varepsilon$:
\begin{lemma}
  For $\varepsilon=\frac{1}{2}$, non-causal knowledge of the erasure does increase the feedback capacity, that is:
  \begin{equation}\label{eq:strong_inequality}
    C_{(2,\infty)}^{\mathrm{fb}}(\frac{1}{2}) < C_{(2,\infty)}^{\mathrm{nc}}(\frac{1}{2}).
  \end{equation}
\end{lemma}
\begin{proof}
  By partially deriving the RHS of \eqref{eq:(2,inf)_upper_bound}, the only critical point in the compact domain $\{(\delta_0,\delta_1,\delta_2)\in\mathbb{R}^3|0\leq\delta_0,\delta_1,\delta_2\leq1\}$ is $\delta\triangleq \delta_0 = \delta_1 = \delta_2$. Substituting $\delta$ into \eqref{eq:(2,inf)_upper_bound} gives the objective of \eqref{eq:(d_inf)_nc_capacity}, so all that is left to show is that the argument which achieves the maximum in \eqref{eq:(d_inf)_nc_capacity} is greater than $\frac{1}{3}$. For $\varepsilon=\frac{1}{2}$, one can show that the maximum of \eqref{eq:(d_inf)_nc_capacity} is obtained at $\frac{1}{3}<\delta<\frac{1}{2}$. This means that the local maximum of \eqref{eq:(d_inf)_nc_capacity} is located outside the maximization domain of \eqref{eq:(2,inf)_upper_bound}. Additional tedious calculations also reveal that \eqref{eq:(d_inf)_nc_capacity} on its boundaries is strictly smaller than its local maximum.
\end{proof}
% 1. A continuously differentiable function gets a maximum in a compact domain
% 2. The function in (19) has a single critical point for \delta_0=\delta_1=\delta_2.
% 3. For \varepsilon = 0.5 the maximum is obtained outside the sub-domain sum of deltas = 1.
% 4. It is possible to show that the function gets only smaller values of the boundary of the domain.

\section{Feedback Capacity of (1,2)-RLL BEC and Future Research}\label{sec:(1,2)-RLL}

    In this section we present the feedback capacity of a BEC with a $(1,2)$-RLL input constraint, $C_{(1,2)}^{\mathrm{fb}}(\varepsilon)$. This is the first example we see in which both $d$ and $k$ constraints are active. Additionally, we discuss possible avenues for future research on this topic.

    \subsection{Feedback Capacity of (1,2)-RLL BEC}\label{subsec:(1,2)-RLL}
    A binary sequence satisfies the $(1,2)$-RLL constraint if every `$1$' is followed by at least one `$0$', but no more than two consecutive `$0$'s are allowed. Graphical representation of the constraint is provided in Fig. \ref{figures:(1,2)_S_graph}. We present a capacity achieving coding scheme and an upper bound based on the Q-graph approach.

    \begin{figure}
        \centering
        \psfrag{A}[c][][1]{$S=0$}
        \psfrag{B}[c][][1]{$S=1$}
        \psfrag{C}[c][][1]{$S=2$}
        \psfrag{D}[c][][1]{$1$}
        \psfrag{E}[c][][1]{$0$}
        \includegraphics[scale=1]{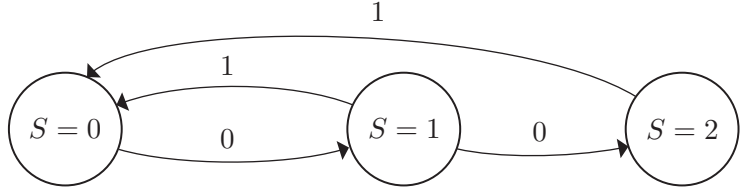}
        \caption{State diagram describing all sequences that can be generated while satisfying the (1,2)-RLL constraint: every `$1$' is followed by a `$0$', and two consecutive `$0$'s are followed by a `$1$'.}\label{figures:(1,2)_S_graph}
    \end{figure}

        The construction of this coding scheme follows closely that of the scheme presented in Section \ref{sec:coding_scheme}. Fig. \ref{figures:(1,2)_coding_scheme_graph} contains a finite state machine we use in this case. The scheme is defined by the FSM in Fig. \ref{figures:(1,2)_coding_scheme_graph} and the following channel input distributions:
        \begin{itemize}
          \item $\Pr(X=0 \mid L=l_1)=\overline{\delta}$.
          \item $\Pr(X=0 \mid L=l_2)=\delta$.
          \item $\Pr(X=0 \mid L=l_3)=0$.
          \item $\Pr(X=0 \mid L=l_4)=1$.
        \end{itemize}
        The partitions of $[0,1)$, i.e., labeling, are not presented because the amount of different labelings increases with time. The next lemma shows that there exists a coding scheme that is defined by the FSM in Fig. \ref{figures:(1,2)_coding_scheme_graph} and the aforementioned input distributions. It also states the conditions under which this scheme does not violate the input constraint.

        \begin{figure}
            \centering
            \psfrag{A}[c][][1]{$l_1$}
            \psfrag{B}[c][][1]{$l_2$}
            \psfrag{C}[c][][1]{$l_3$}
            \psfrag{D}[c][][1]{$l_4$}
            \psfrag{E}[c][][1]{$0$}
            \psfrag{F}[c][][1]{$1$}
            \psfrag{G}[c][][1]{$?$}
            \psfrag{H}[c][][1]{$0/?/1$}
            \psfrag{I}[c][][1]{$0/?/1$}
            \includegraphics[scale=1]{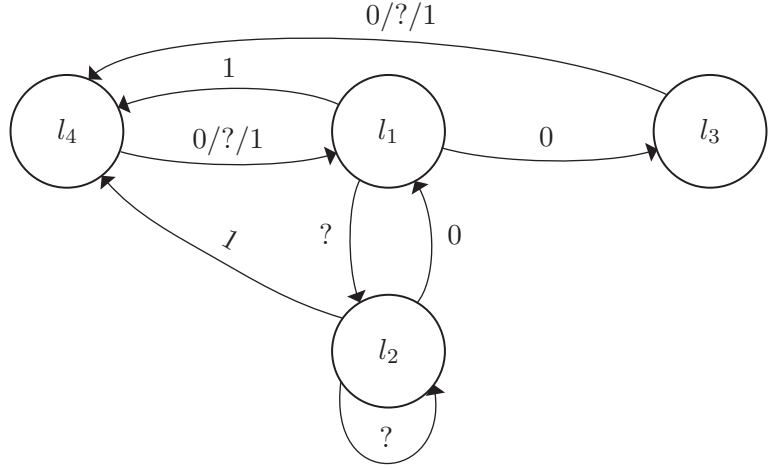}
            \caption{FSM which defines the coding scheme. The nodes describe the instantaneous labelling that is used by the encoder. Edges correspond to channel outputs. In node $1$ $\Pr(X=0)=\overline{\delta}$, in node $2$ $\Pr(X=0)=\delta$, in node $3$ $\Pr(X=0)=0$ and in node $4$ $\Pr(X=0)=1$.}\label{figures:(1,2)_coding_scheme_graph}
        \end{figure}

        \begin{lemma}\label{lemma:(1,2)_feasibility}
            For $\frac{1}{2}\leq\delta\leq\frac{2}{3}$, the coding scheme in Fig \ref{figures:(1,2)_coding_scheme_graph} does not violate the $(1,2)$-RLL input constraint and achieves:
            \begin{equation}\label{eq:(1,2)_lower_bound}
              R=\frac{H_2(\delta)} {\frac{1}{1-\varepsilon}+\overline{\varepsilon}+\overline{\delta}}.
            \end{equation}
        \end{lemma}
        The proof of Lemma \ref{lemma:(1,2)_feasibility} is presented in Appendix \ref{sec:proofs_for_(1,2)-RLL}. Thus, a lower bound on the feedback capacity is:
        \begin{equation*}\label{}
            C_{(1,2)}^{\mathrm{fb}}(\varepsilon)\geq\max_{\frac{1}{3}\leq\delta\leq\frac{1}{2}} \frac{H_2(\delta)} {\frac{1}{1-\varepsilon}+\overline{\varepsilon}+\delta}.
        \end{equation*}

        For the upper bound, we use the same Q-graph technique from Section \ref{sec:(2,infty)-RLL} in Theorem \ref{thm:Oron_Q_graph}. This time, the coding scheme graph presented in Fig. \ref{figures:(1,2)_coding_scheme_graph} is chosen as our Q-graph. Calculating $\sup_{p(x \mid s,q)}I(X;Y \mid Q)$ we get:
        \begin{equation}\label{(eq:1,2)_1st_upper_bound}
            C_{(1,2)}^{\mathrm{fb}}\leq \max_{0\leq\delta_1,\delta_2\leq1} \frac{\overline{\varepsilon}^2H_2(\delta_1)+\varepsilon\overline{\varepsilon}H_2(\delta_2)} {1+\overline{\varepsilon}+\overline{\varepsilon}\delta_1+\varepsilon\overline{\varepsilon}\delta_2}.
        \end{equation}
        The following lemma shows that the upper and lower bounds coincide:
        \begin{lemma}\label{lemma:(1,2)_bounds_coincide}
        The feedback capacity of the $(1,2)$-RLL input constrained BEC is upper bounded by:
            \begin{equation*}
                C_{(1,2)}^{\mathrm{fb}}\leq\max_{\frac{1}{3}\leq\delta\leq\frac{1}{2}} \frac{H_2(\delta)} {\frac{1}{1-\varepsilon}+\overline{\varepsilon}+\delta}.
            \end{equation*}
        \end{lemma}
        This proof also appears in Appendix \ref{sec:proofs_for_(1,2)-RLL}. This completes the derivation of the capacity of the $(1,2)$-RLL input constrained BEC.

    \subsection{Future Research}

        As indicated by the $(1,2)$-RLL example, the most logical course of future research is to study the feedback capacity of the general $(d,k)$-RLL input constrained BEC for any natural $d<k$. Our method of tackling the various input constraints discussed in this paper consisted of first running numerical evaluations of the equivalent DP problems, and then trying to draw conclusions as to the capacity achieving coding scheme. However, it is important to notice that the amount of variables we need to numerically evaluate grows linearly with the parameters $d$ and $k$. Thus, this somewhat naive approach will probably not suffice to find the capacity expression for the general case. The feedback capacity of the general $(d,k)$-RLL input constrained BEC is still open, in particular for $d=1,k\geq3$ and for $2\leq d<k$.

        We have also invested efforts in solving the second famous family of the $(d,k)$-RLL constraints, the $(d,\infty)$-RLL constraints. As illustrated in the previous section, the non causal capacity is not a tight upper bound on the capacity, so alternative methods to the ones presented in this paper should be applied to tackle these constraints. Based on numerical experiments, we tend to believe that the capacity in this case will be an optimization over more than one parameter.

\bibliography{Paper3}
\bibliographystyle{IEEEtran}

\begin{appendices}

\section{Equality of the Bounds}\label{sec:equality_of_bounds}
 In this appendix, we prove Lemma \ref{lemma:main_lemma}. It states that the lower and upper bounds, calculated in Sections \ref{sec:coding_scheme} and \ref{sec:upper_bound}, respectively, are equal. In addition, it shows that $(\delta_0,\ldots,\delta_{k-1})$ that maximize $R_\varepsilon\left(\delta_0, \ldots ,\delta_{k-1}\right)$ are connected to each other in a series of equations that allow us to compute $(\delta_0,\ldots,\delta_{k-2})$ once the maximizing $\delta_{k-1}$ is known.

 Denote:
\begin{equation}\label{eq:domain1}
  D_1 = \left\{\left(\delta_0, \ldots ,\delta_{k-1}\right) \in \mathbb{R}^k \, | \, 0 \leq \delta_0 ,\ldots, \delta_{k-1} \leq 1 \right\} ,
\end{equation}
\begin{equation}\label{eq:domain2}
  D_2 = \left\{\left(\delta_0, \ldots ,\delta_{k-1}\right) \in \mathbb{R}^k \, | \, 0 \leq \delta_0 ,\ldots, \delta_{k-1} \leq \frac{1}{2} \right\}.
\end{equation}
 Define $\vec{\delta}^\ast = (\delta_0^\ast,\ldots,\delta_{k-1}^\ast) \stackrel{\text{def}}{=} \argmax_{D_1} R_\varepsilon\left(\delta_0, \ldots ,\delta_{k-1}\right) $. We aim to show that $\vec{\delta}^\ast\in D_2$. The proof is spread across several lemmas, which show the following:
 \begin{itemize}
   \item In Lemma \ref{lemma:delta's_relations} we prove that $\nabla R_\varepsilon(\vec{\delta})=0\iff\vec{\delta}$ satisfies Eqs. \eqref{eq:delta's}. We also show that Eqs. \eqref{eq:delta's} imply that $\delta_0 \geq \ldots \geq \delta_{k-1}$.
   \item Lemma \ref{lemma:delta_0<1/2} proves that for any $(\delta_1,\ldots,\delta_{k-1})\in D_1$ there exists a unique $0\leq\delta_0(\delta_1,\ldots,\delta_{k-1})\leq\frac{1}{2}$, which is denoted by $\delta_0^\ast$, such that $\frac{\partial R_\varepsilon\left(\delta_0, \ldots ,\delta_{k-1}\right)}{\partial \delta_0}\Big|_{\delta_0=\delta_0^\ast} = 0$. Lemmas \ref{lemma:delta's_relations} and \ref{lemma:delta_0<1/2} together show that there exists a unique $\vec{\delta}\in D_2$ such that $\nabla R_\varepsilon(\vec{\delta})=0$.
   \item Lemma \ref{lemma:KKT} proves that $R_\varepsilon(\delta_0,\ldots,\delta_{k-1})$ has no maximum on the boundary of $D_1$, and hence, $\vec{\delta}^\ast\in D_2$.
 \end{itemize}

 To simplify notation, for $k>l$ we define $\prod_{m=k}^{l}(\cdot) \stackrel{\text{def}}{=} 1$ and $\sum_{i=k}^{l}(\cdot) \stackrel{\text{def}}{=} 0$.

 \begin{lemma}\label{lemma:delta's_relations}
   A $k$-tuple $\vec{\delta}=(\delta_0,\ldots,\delta_{k-1})\in D_1$ satisfies $\nabla R_\varepsilon(\vec{\delta})=0$ if and only if
   \begin{equation*}
  \delta_j = \frac{\delta_{j+1}}{\delta_{j+1}+\overline{\delta}_{j+1}\left( \frac{\overline{\delta}_{j+1}}{\overline{\delta}_{j+2}} \right)^{\overline{\varepsilon}}} \quad j=0,1,\ldots,k-2 ,
  \end{equation*}
  where we define $\overline{\delta}_k=1$. In addition $\delta_0\geq\ldots\geq\delta_{k-1}$.
 \end{lemma}

 \begin{proof}
 First prove that $\nabla R_\varepsilon(\vec{\delta})=0$ if and only if the following relation holds:
 \begin{equation*}
     \log\left( \frac{\overline{\delta}_j}{\delta_j} \right) = \log\left( \frac{\overline{\delta}_{j+1}}{\delta_{j+1}} \right) +\overline{\varepsilon}\log\left( \frac{\overline{\delta}_{j+1}}{\overline{\delta}_{j+2}} \right)
     \quad\quad j=0,1,\ldots,k-2.
   \end{equation*}
   Denote:
   \begin{equation}\label{eq:N_and_D_notation}
     N = \displaystyle\sum_{i=0}^{k-1}\left(\overline{\varepsilon}^{i+1}H_2(\delta_i)\prod_{m=0}^{i-1}\delta_m\right) \quad,\quad
     D = 1+\displaystyle\sum_{i=0}^{k-1}\left(\overline{\varepsilon}^{i+1}\prod_{m=0}^{i}\delta_m\right).
   \end{equation}
   So
   \begin{equation*}
     R_\varepsilon\left(\delta_0, \ldots ,\delta_{k-1}\right) = \frac{N}{D}
   \end{equation*}
   and,
   \begin{equation*}
     \frac{\partial R_\varepsilon(\delta_0, \ldots ,\delta_{k-1})}{\partial \delta_0} = \frac{\frac{\partial N}{\partial \delta_0} D - N\frac{\partial D}{\partial\delta_0}}{D^2}.
   \end{equation*}
   We write the partial derivative $\frac{\partial R_\varepsilon(\delta_0\ldots\delta_{k-1})}{\partial\delta_j}$ for $j=0,1,\ldots,k-1$ using the notations introduced in \eqref{eq:N_and_D_notation}:
   \begin{equation}\label{eq:general_partial_derivative}
     \frac{\partial R_\varepsilon(\delta_0\ldots\delta_{k-1})}{\partial\delta_j} = \frac{\left( \overline{\varepsilon}^{j+1}\displaystyle\prod_{m=0}^{j-1}\delta_m\log\left( \frac{\overline{\delta}_j}{\delta_j} \right)
     + \displaystyle\sum_{i=j+1}^{k-1}\overline{\varepsilon}^{i+1}H_2(\delta_i)\displaystyle\prod_{\substack{m=0 \\ m\neq j}}^{i-1}\delta_m \right)D
     - N\displaystyle\sum_{i=j}^{k-1}\overline{\varepsilon}^{i+1}\displaystyle\prod_{\substack{m=0 \\ m\neq j}}^i\delta_m}{D^2}.
   \end{equation}
   We will prove the lemma using an inductive argument starting from $\delta_{k-1}$ and working our way back to $\delta_0$.\\
   \underline{Base case}: by simplifying the equation $\frac{\partial R_\varepsilon}{\delta_{k-1}}=0$ we immediately get:
   \begin{equation}\label{eq:N_{k-1}}
     N = D\log\left( \frac{\overline{\delta}_{k-1}}{\delta_{k-1}} \right).
   \end{equation}
   Note that we arrive at \eqref{eq:N_{k-1}} by dividing both sides of the equation by $\prod_{m=0}^{k-2}\delta_m$. We know that this is allowed since for any $j=0,\ldots,k-1$ we have that $\lim_{\delta_j\to 0^+}\frac{\partial R_\varepsilon(\delta_0\ldots\delta_{k-1})}{\partial\delta_j} = \infty$. This means that if $\nabla R_\varepsilon(\delta_0\ldots\delta_{k-1}) = 0$ then $\delta_j \neq 0$ for all $j=0,\ldots,k-1$.

   Next we write the equation $\frac{\partial R_\varepsilon}{\partial\delta_{k-2}}=0$ and substitute $N$ using \eqref{eq:N_{k-1}}:
   \begin{align*}
     0 &= \frac{\left( \overline{\varepsilon}^{k-1}\displaystyle\prod_{m=0}^{k-3}\delta_m\log\left( \frac{\overline{\delta}_{k-2}}{\delta_{k-2}} \right)
     + \overline{\varepsilon}^{k} H_2(\delta_{k-1})\displaystyle\prod_{m=0}^{k-3}\delta_m \right)D
     - D\log\left( \frac{\overline{\delta}_{k-1}}{\delta_{k-1}}\right)\left[ \overline{\varepsilon}^{k-1}\displaystyle\prod_{m=0}^{k-3}\delta_m+\overline{\varepsilon}^{k}
     \displaystyle\prod_{m=0}^{k-3}\delta_m\delta_{k-1} \right]}{D^2}  \\
      &= \log\left( \frac{\overline{\delta}_{k-2}}{\delta_{k-2}} \right) - \overline{\varepsilon}\delta_{k-1}\log(\delta_{k-1})-\overline{\varepsilon}\overline{\delta}_{k-1}\log(\overline{\delta}_{k-1})
     - \log\left( \frac{\overline{\delta}_{k-1}}{\delta_{k-1}}\right) - \overline{\varepsilon}\delta_{k-1}\log(\overline{\delta}_{k-1}) + \overline{\varepsilon}\delta_{k-1}\log(\delta_{k-1}). \\
   \end{align*}
   So
   \begin{equation}\label{eq:base_case}
     \log\left( \frac{\overline{\delta}_{k-2}}{\delta_{k-2}} \right) = \log\left( \frac{\overline{\delta}_{k-1}}{\delta_{k-1}}\right) + \overline{\varepsilon}\log\left( \frac{\overline{\delta}_{k-1}}{1} \right)
   \end{equation}
   and the base case is proven.\\
   \underline{Inductive step}: we assume that the claim holds for $\delta_{k-2},\delta_{k-3},\ldots,\delta_{j+1}$ and we will now prove it for $\delta_j$. Substituting \eqref{eq:base_case} into \eqref{eq:N_{k-1}} we get:
   \begin{equation}\label{eq:N_2}
     N = D\left( \log\left( \frac{\overline{\delta}_{k-2}}{\delta_{k-2}} \right) - \overline{\varepsilon}\log(\overline{\delta}_{k-1}) \right).
   \end{equation}
   We start by writing the equation
   \begin{align*}
     0 &= \frac{\partial R_\varepsilon(\delta_0,\ldots,\delta_{k-1})}{\partial\delta_j} \\
      &= \frac{\left( \overline{\varepsilon}^{j+1}\displaystyle\prod_{m=0}^{j-1}\delta_m\log\left( \frac{\overline{\delta}_j}{\delta_j} \right)
     + \displaystyle\sum_{i=j+1}^{k-1}\overline{\varepsilon}^{i+1}H_2(\delta_i)\displaystyle\prod_{\substack{m=0 \\ m\neq j}}^{i-1}\delta_m \right)D
     - N\displaystyle\sum_{i=j}^{k-1}\overline{\varepsilon}^{i+1}\displaystyle\prod_{\substack{m=0 \\ m\neq j}}^i\delta_m}{D^2}.
   \end{align*}
   We can divide by $\varepsilon^{j+1}\prod_{m=0}^{j-1}\delta_m$ and use \eqref{eq:N_2} to get:
   \begin{equation*}
     0 = \left( \log\left( \frac{\overline{\delta}_j}{\delta_j} \right)
     + \displaystyle\sum_{i=j+1}^{k-1}\overline{\varepsilon}^{i-j}H_2(\delta_i)\displaystyle\prod_{m=j+1}^{i-1}\delta_m \right)D
     - D\left( \log\left( \frac{\overline{\delta}_{k-2}}{\delta_{k-2}} \right) - \overline{\varepsilon}\log(\overline{\delta}_{k-1}) \right)
     \displaystyle\sum_{i=j}^{k-1}\overline{\varepsilon}^{i-j}\displaystyle\prod_{m=j+1}^i\delta_m.
   \end{equation*}
   Next we use the definition of the binary entropy function to replace $H_2(\delta_{k-2}),H_2(\delta_{k-1})$ with an explicit expression:
   \begin{align*}
     0 &= \log\left( \frac{\overline{\delta}_j}{\delta_j} \right) + \displaystyle\sum_{i=j+1}^{k-3}\overline{\varepsilon}^{i-j}H_2(\delta_i)\displaystyle\prod_{m=j+1}^{i-1}\delta_m
     - \overline{\varepsilon}^{k-2-j}\displaystyle\prod_{m=j+1}^{k-3}\delta_m\delta_{k-2}\log(\delta_{k-2}) \\
     &- \overline{\varepsilon}^{k-2-j}\displaystyle\prod_{m=j+1}^{k-3}\delta_m\overline{\delta}_{k-2}\log(\overline{\delta}_{k-2})
     - \overline{\varepsilon}^{k-1-j}\displaystyle\prod_{m=j+1}^{k-2}\delta_m\delta_{k-1}\log(\delta_{k-1}) \\
     &- \overline{\varepsilon}^{k-1-j}\displaystyle\prod_{m=j+1}^{k-2}\delta_m\overline{\delta}_{k-1}\log(\overline{\delta}_{k-1})
     - \log\left( \frac{\overline{\delta}_{k-2}}{\delta_{k-2}} \right) \displaystyle\sum_{i=j}^{k-3}\overline{\varepsilon}^{i-j}\displaystyle\prod_{m=j+1}^{i}\delta_m \\
     &- \overline{\varepsilon}^{k-2-j}\displaystyle\prod_{m=j+1}^{k-2}\delta_m\log(\overline{\delta}_{k-2})
     + \overline{\varepsilon}^{k-2-j}\displaystyle\prod_{m=j+1}^{k-3}\delta_m\delta_{k-2}\log(\delta_{k-2}) \\
     &- \log\left( \frac{\overline{\delta}_{k-2}}{\delta_{k-2}} \right)\overline{\varepsilon}^{k-1-j}\displaystyle\prod_{m=j+1}^{k-1}\delta_m
     + \log(\overline{\delta}_{k-1})\displaystyle\sum_{i=j}^{k-3}\overline{\varepsilon}^{i-j+1}\displaystyle\prod_{m=j+1}^i\delta_m
     + \log(\overline{\delta}_{k-1})\overline{\varepsilon}^{k-1-j}\displaystyle\prod_{m=j+1}^{k-2}\delta_m \\
     &+ \log(\overline{\delta}_{k-1})\overline{\varepsilon}^{k-j}\displaystyle\prod_{m=j+1}^{k-1}\delta_m.
   \end{align*}
   Recall that $\overline{\delta} = 1- \delta$, so we can simplify this expression:
   \begin{align}\label{1}
     0 &= \log\left( \frac{\overline{\delta}_j}{\delta_j} \right)+ \displaystyle\sum_{i=j+1}^{k-3}\overline{\varepsilon}^{i-j}H_2(\delta_i)\displaystyle\prod_{m=j+1}^{i-1}\delta_m
     - \overline{\varepsilon}^{k-2-j}\displaystyle\prod_{m=j+1}^{k-3}\delta_m\log(\overline{\delta}_{k-2}) \nonumber \\
      &+ \overbrace{ \log\left( \frac{\overline{\delta}_{k-1}}{\delta_{k-1}} \right)\overline{\varepsilon}^{k-1-j}\displaystyle\prod_{m=j+1}^{k-1}\delta_m }^{(\ast)}
      - \log\left( \frac{\overline{\delta}_{k-2}}{\delta_{k-2}} \right) \displaystyle\sum_{i=j}^{k-3}\overline{\varepsilon}^{i-j}\displaystyle\prod_{m=j+1}^{i}\delta_m  \\
      &- \overbrace { \log\left( \frac{\overline{\delta}_{k-2}}{\delta_{k-2}} \right)\overline{\varepsilon}^{k-1-j}\displaystyle\prod_{m=j+1}^{k-1}\delta_m }^{(\ast)}
      + \log(\overline{\delta}_{k-1})\displaystyle\sum_{i=j}^{k-3}\overline{\varepsilon}^{i-j+1}\displaystyle\prod_{m=j+1}^i\delta_m
      + \overbrace{ \log(\delta_{k-1})\overline{\varepsilon}^{k-j}\displaystyle\prod_{m=j+1}^{k-1}\delta_m }^{(\ast)}. \nonumber
   \end{align}
   The three expression marked with $(\ast)$ cancel each other as a result of \eqref{eq:base_case}. Now we will use the induction assumption again by substituting
   \begin{equation*}
     \log\left( \frac{\overline{\delta}_{k-2}}{\delta_{k-2}} \right) = \log\left( \frac{\overline{\delta}_{k-3}}{\delta_{k-3}} \right)
     - \overline{\varepsilon}\log\left( \frac{\overline{\delta}_{k-2}}{\overline{\delta}_{k-1}} \right),
   \end{equation*}
   so
   \begin{align*}
     0 &= \log\left( \frac{\overline{\delta}_j}{\delta_j} \right)+ \displaystyle\sum_{i=j+1}^{k-4}\overline{\varepsilon}^{i-j}H_2(\delta_i)\displaystyle\prod_{m=j+1}^{i-1}\delta_m
     - \overbrace{ \overline{\varepsilon}^{k-3-j}\displaystyle\prod_{m=j+1}^{k-4}\delta_m\delta_{k-3}\log(\delta_{k-3}) }^{(\ast)}  \\
       &- \overline{\varepsilon}^{k-3-j}\displaystyle\prod_{m=j+1}^{k-4}\delta_m\overline{\delta}_{k-3}\log(\overline{\delta}_{k-3})
       - \overbrace { \overline{\varepsilon}^{k-2-j}\displaystyle\prod_{m=j+1}^{k-3}\delta_m\log(\overline{\delta}_{k-2}) }^{(\ast)}
       - \log\left( \frac{\overline{\delta}_{k-3}}{\delta_{k-3}} \right) + \overline{\varepsilon}\log(\overline{\delta}_{k-2})  \\
       &- \overbrace{ \overline{\varepsilon}\log(\overline{\delta}_{k-1}) }^{(\ast)} - \overline{\varepsilon}\delta_{j+1}\log\left( \frac{\overline{\delta}_{k-3}}{\delta_{k-3}} \right)
       + \overline{\varepsilon}^2\delta_{j+1}\log(\overline{\delta}_{k-2}) - \overbrace{ \overline{\varepsilon}^2\delta_{j+1}\log(\overline{\delta}_{k-1}) }^{(\ast)} \\
       &- \log\left( \frac{\overline{\delta}_{k-3}}{\delta_{k-3}} \right) \displaystyle\sum_{i=j+2}^{k-4}\overline{\varepsilon}^{i-j}\displaystyle\prod_{m=j+1}^{i}\delta_m
       + \log(\overline{\delta}_{k-2})\displaystyle\sum_{i=j+2}^{k-4}\overline{\varepsilon}^{i-j+1}\displaystyle\prod_{m=j+1}^{i}\delta_m
       - \overbrace{ \log(\overline{\delta}_{k-1})\displaystyle\sum_{i=j+2}^{k-4}\overline{\varepsilon}^{i-j+1}\displaystyle\prod_{m=j+1}^{i}\delta_m }^{(\ast)} \\
       &- \log(\overline{\delta}_{k-3})\overline{\varepsilon}^{k-3-j}\displaystyle\prod_{m=j+1}^{k-3}\delta_m
       + \overbrace { \log(\delta_{k-3})\overline{\varepsilon}^{k-3-j}\displaystyle\prod_{m=j+1}^{k-3}\delta_m }^{(\ast)}
       + \overbrace { \log(\overline{\delta}_{k-2})\overline{\varepsilon}^{k-2-j}\displaystyle\prod_{m=j+1}^{k-3}\delta_m }^{(\ast)} \\
       &- \overbrace { \log(\overline{\delta}_{k-1})\overline{\varepsilon}^{k-2-j}\displaystyle\prod_{m=j+1}^{k-3}\delta_m }^{(\ast)}
       + \overbrace { \log(\overline{\delta}_{k-1})\displaystyle\sum_{i=j}^{k-3}\overline{\varepsilon}^{i-j+1}\displaystyle\prod_{m=j+1}^i\delta_m }^{(\ast)}.
   \end{align*}
   All expressions marked with $(\ast)$ cancel each other out. Again, using $\overline{\delta} = 1-\delta$ we can simplify and arrive at:
   \begin{align}\label{2}
     0 &= \log\left( \frac{\overline{\delta}_j}{\delta_j} \right)+ \displaystyle\sum_{i=j+1}^{k-4}\overline{\varepsilon}^{i-j}H_2(\delta_i)\displaystyle\prod_{m=j+1}^{i-1}\delta_m
     - \overline{\varepsilon}^{k-3-j}\displaystyle\prod_{m=j+1}^{k-4}\delta_m\log(\overline{\delta}_{k-3}) \nonumber \\
     &- \log\left( \frac{\overline{\delta}_{k-3}}{\delta_{k-3}} \right) \displaystyle\sum_{i=j}^{k-4}\overline{\varepsilon}^{i-j}\displaystyle\prod_{m=j+1}^{i}\delta_m
      + \log(\overline{\delta}_{k-2})\displaystyle\sum_{i=j}^{k-4}\overline{\varepsilon}^{i-j+1}\displaystyle\prod_{m=j+1}^i\delta_m.
   \end{align}
   When we compare \eqref{2} to \eqref{1} we see a pattern emerging. Continuing to perform these substitutions we reach:
   \begin{align*}
     0 &= \log\left( \frac{\overline{\delta}_j}{\delta_j} \right) -\overline{\varepsilon}\delta_{j+1}\log(\delta_{j+1})
     - \overline{\varepsilon}\overline{\delta}_{j+1}\log(\overline{\delta}_{j+1}) - \overline{\varepsilon}^2\delta_{j+1}\log(\overline{\delta}_{j+2}) \\
      &-  \log\left( \frac{\overline{\delta}_{j+2}}{\delta_{j+2}} \right)(1+\overline{\varepsilon}\delta_{j+1})
      + \log(\overline{\delta}_{j+2})(\overline{\varepsilon}+\overline{\varepsilon}^2\delta_{j+1}).
   \end{align*}
   Performing the final substitution and simplifying further we get:
   \begin{align*}
     0 &= \log\left( \frac{\overline{\delta}_j}{\delta_j} \right) - \overline{\varepsilon}\delta_{j+1}\log(\delta_{j+1})
     - \overline{\varepsilon}\overline{\delta}_{j+1}\log(\overline{\delta}_{j+1}) - \overline{\varepsilon}^2\delta_{j+1}\log(\overline{\delta}_{j+2})
     -  \log\left( \frac{\overline{\delta}_{j+1}}{\delta_{j+1}} \right) \\
      &+ \overline{\varepsilon}\log(\overline{\delta}_{j+2}) -  \overline{\varepsilon}\log(\overline{\delta}_{j+3})
      -\overline{\varepsilon}\delta_{j+1}\log(\overline{\delta}_{j+1}) + \overline{\varepsilon}\delta_{j+1}\log(\delta_{j+1})
      + \overline{\varepsilon}^2\delta_{j+1}\log(\overline{\delta}_{j+2})  \\
      &- \overline{\varepsilon}^2\delta_{j+1}\log(\overline{\delta}_{j+3}) + \overline{\varepsilon}\log(\overline{\delta}_{j+3})
      + \overline{\varepsilon}^2\delta_{j+1}\log(\overline{\delta}_{j+3}),
   \end{align*}
   and, finally, we arrive at:
   \begin{equation}\label{eq:inductive_step}
     \log\left( \frac{\overline{\delta}_j}{\delta_j} \right) = \log\left( \frac{\overline{\delta}_{j+1}}{\delta_{j+1}} \right) +\overline{\varepsilon}\log\left( \frac{\overline{\delta}_{j+1}}{\delta_{j+2}} \right).
   \end{equation}

   Now we will use induction again to prove that
   \begin{equation*}
  \delta_j = \frac{\delta_{j+1}}{\delta_{j+1}+\overline{\delta}_{j+1}\left( \frac{\overline{\delta}_{j+1}}{\overline{\delta}_{j+2}} \right)^{\overline{\varepsilon}}} \quad j=0,1,\ldots,k-2
  \end{equation*}
  and that $\delta_0\geq\delta_1\geq\ldots\geq\delta_{k-1}$.

   \underline{Base case}: we will start by showing that $\delta_{k-2}\geq\delta_{k-1}$. In \eqref{eq:base_case} we have that
   \begin{equation*}
     \log\left( \frac{\overline{\delta}_{k-2}}{\delta_{k-2}} \right) = \log\left( \frac{\overline{\delta}_{k-1}}{\delta_{k-1}}\right) + \overline{\varepsilon}\log\left( \overline{\delta}_{k-1} \right).
   \end{equation*}
   By rearranging this equation we get:
   \begin{equation*}
     \frac{\overline{\delta}_{k-2}}{\delta_{k-2}}=\frac{\overline{\delta}_{k-1}^{1+\overline{\varepsilon}}}{\delta_{k-1}} ,
   \end{equation*}
   which means that:
   \begin{align*}
     \delta_{k-2} =& \frac{1}{1+\frac{\overline{\delta}_{k-1}^{1+\overline{\varepsilon}}}{\delta_{k-1}}} \\
      =& \frac{\delta_{k-1}}{\delta_{k-1}+\overline{\delta}_{k-1}^{1+\overline{\varepsilon}}}.
   \end{align*}
   Note that assuming $\delta_{k-1}>0$:
   \begin{equation*}
     \delta_{k-1}+\overline{\delta}_{k-1}^{1+\overline{\varepsilon}}<1 \iff \overline{\delta}_{k-1}^{\overline{\varepsilon}}<1
   \end{equation*}
   and the right hand side of this equivalence surely holds (under said assumption). We have proven the base case.\\
   \underline{Inductive step}: we assume that the claim holds for $\delta_{k-2},\delta_{k-3},\ldots,\delta_{j+1}$ and we will now prove it for $\delta_j$. In \eqref{eq:inductive_step} we have:
   \begin{equation*}
     \log\left( \frac{\overline{\delta}_j}{\delta_j} \right) = \log\left( \frac{\overline{\delta}_{j+1}}{\delta_{j+1}} \right) +\overline{\varepsilon}\log\left( \frac{\overline{\delta}_{j+1}}{\delta_{j+2}} \right).
   \end{equation*}
   Following the same steps as in the base case we arrive at:
   \begin{equation*}
     \delta_j = \frac{\delta_{j+1}}{\delta_{j+1}+\overline{\delta}_{j+1}\left( \frac{\overline{\delta}_{j+1}}{\overline{\delta}_{j+2}} \right)^{\overline{\varepsilon}} }.
   \end{equation*}
   Now,
   \begin{equation*}
     \delta_{j+1}+\overline{\delta}_{j+1}\left( \frac{\overline{\delta}_{j+1}}{\overline{\delta}_{j+2}} \right)^{\overline{\varepsilon}}<1 \iff \left( \frac{\overline{\delta}_{j+1}}{\overline{\delta}_{j+2}} \right)^{\overline{\varepsilon}}<1,
   \end{equation*}
   and we know that the right hand side of the equivalence holds thanks to the induction assumption (note that $\delta_{j+1}>\delta_{j+2} \implies \overline{\delta}_{j+1}<\overline{\delta}_{j+2}$).
 \end{proof}

 This lemma shows that for any $(\delta_1,\ldots,\delta_{k-1})$ satisfying $0\leq \delta_1,\ldots,\delta_{k-1} \leq 1$ there exists $0\leq\overline{\delta}_0\leq\frac{1}{2}$ for which $\frac{\partial R_\varepsilon(\delta_0,\ldots,\delta_{k-1})}{\partial\delta_0}\big|_{\delta_0=\overline{\delta}_0} = 0$ and that this $\overline{\delta}_0$ is unique.
 \begin{lemma}\label{lemma:delta_0<1/2}
 This lemma has two parts:
 \begin{enumerate}
   \item For any $(\delta_1,\ldots,\delta_{k-1})$ satisfying $0\leq \delta_1,\ldots,\delta_{k-1} \leq 1$ there exists $0\leq\delta_0(\delta_1,\ldots,\delta_{k-1})<\frac{1}{2}$, which we denote $\bar{\delta}_0$, such that:
       \begin{equation*}
         \frac{\partial R_\varepsilon\left(\delta_0, \ldots ,\delta_{k-1}\right)}{\partial \delta_0}\Big|_{\delta_0=\bar{\delta}_0} = 0.
       \end{equation*}
   \item The partial derivative $\frac{\partial R_\varepsilon\left(\delta_0, \ldots ,\delta_{k-1}\right)}{\partial \delta_0}$ is monotonic non-increasing in $\delta_0$.
 \end{enumerate}
 \end{lemma}
 \begin{proof}
 We calculate $\frac{\partial R_\varepsilon(\delta_0,\ldots,\delta_{k-1})}{\partial\delta_0}$ and show that:
 \begin{equation}\label{eq:delta_to_0}
     \lim_{\delta_0 \to 0^+} \frac{\partial R_\varepsilon\left(\delta_0, \ldots ,\delta_{k-1}\right)}{\partial \delta_0} > 0
   \end{equation}
    and
    \begin{equation}\label{eq:delta_1/2}
      \frac{\partial R_\varepsilon\left(\delta_0, \ldots ,\delta_{k-1}\right)}{\partial \delta_0}\Big|_{\delta_0=\frac{1}{2}} < 0.
    \end{equation}
    Since the partial derivative is a continuous function of $\delta_0$ we can use the intermediate value theorem to prove the first part of the lemma. Recall that:
   \begin{equation*}
     \frac{\partial R_\varepsilon(\delta_0, \ldots ,\delta_{k-1})}{\partial \delta_0} = \frac{\frac{\partial N}{\partial \delta_0} D - N\frac{\partial D}{\partial\delta_0}}{D^2}.
   \end{equation*}
   First note that $D^2>0$. This means that we only need to determine the sign of $\frac{\partial N}{\partial \delta_0} D - N\frac{\partial D}{\partial\delta_0}$ as $\delta_0\to0^+$ and for $\delta_0=\frac{1}{2}$ to prove that \eqref{eq:delta_to_0} and \eqref{eq:delta_1/2} hold. Since the expression $\frac{\partial R_\varepsilon}{\partial\delta_0}$ is a long one, we will divide it into two parts:
   \begin{equation}\label{notation:N'D}
     \frac{\partial N}{\partial \delta_0} D =
     \left(\overline{\varepsilon} \log\left( \frac{\overline{\delta}_0}{\delta_0}\right) + \displaystyle\sum_{i=1}^{k-1}\overline{\varepsilon}^{i+1}H_2(\delta_i)\displaystyle\prod_{m=1}^{i-1}\delta_m \right)
     \left( 1+\displaystyle\sum_{i=0}^{k-1}\overline{\varepsilon}^{i+1}\displaystyle\prod_{m=0}^{i}\delta_m \right) ,
   \end{equation}
   \begin{equation}\label{notation:ND'}
     N\frac{\partial D}{\partial\delta_0} = \left( \displaystyle\sum_{i=0}^{k-1}\overline{\varepsilon}^{i+1}H_2(\delta_i)\displaystyle\prod_{m=0}^{i-1}\delta_m \right)
     \left( \displaystyle\sum_{i=0}^{k-1}\overline{\varepsilon}^{i+1}\displaystyle\prod_{m=1}^{i}\delta_m \right) .
   \end{equation}
   Simplifying $\frac{\partial N}{\partial \delta_0} D - N\frac{\partial D}{\partial\delta_0}$ we get:
     \begin{align}
        \frac{\partial N}{\partial \delta_0} D - N\frac{\partial D}{\partial\delta_0} &= \overline{\varepsilon}\log\left( \frac{\overline{\delta}_0}{\delta_0} \right) + \log\left( \frac{\overline{\delta}_0}{\delta_0} \right)\displaystyle\sum_{i=0}^{k-1}\overline{\varepsilon}^{i+2}\displaystyle\prod_{m=0}^{i}\delta_m
        + \displaystyle\sum_{i=1}^{k-1}\overline{\varepsilon}^{i+1}H_2(\delta_i)\displaystyle\prod_{m=1}^{i-1}\delta_m \nonumber \\
          &+ {\left( \displaystyle\sum_{i=1}^{k-1}\overline{\varepsilon}^{i+1}H_2(\delta_i)\displaystyle\prod_{m=1}^{i-1}\delta_m \right)
          \left( \displaystyle\sum_{i=0}^{k-1}\overline{\varepsilon}^{i+1}\displaystyle\prod_{m=0}^{i}\delta_m \right) }
          -\overline{\varepsilon}H_2(\delta_0)\displaystyle\sum_{i=0}^{k-1}\overline{\varepsilon}^{i+1}\displaystyle\prod_{m=1}^{i}\delta_m \nonumber \\
          &- {\left( \displaystyle\sum_{i=1}^{k-1}\overline{\varepsilon}^{i+1}H_2(\delta_i)\displaystyle\prod_{m=1}^{i-1}\delta_m \right)
          \left( \displaystyle\sum_{i=0}^{k-1}\overline{\varepsilon}^{i+1}\displaystyle\prod_{m=0}^{i}\delta_m \right) } \nonumber \\
          &= \log(\overline{\delta}_0)\left[ \overline{\varepsilon}+\displaystyle\sum_{i=0}^{k-1}\overline{\varepsilon}^{i+2}\displaystyle\prod_{m=0}^{i}\delta_m
          + \overline{\delta}_0\displaystyle\sum_{i=0}^{k-1}\overline{\varepsilon}^{i+2}\displaystyle\prod_{m=1}^{i}\delta_m \right] \nonumber \\
          &+ \log(\delta_0)\left[ -\overline{\varepsilon} - {\displaystyle\sum_{i=0}^{k-1}\overline{\varepsilon}^{i+2}\displaystyle\prod_{m=0}^{i}\delta_m} + {\delta_0\displaystyle\sum_{i=0}^{k-1}\overline{\varepsilon}^{i+2}\displaystyle\prod_{m=1}^{i}\delta_m} \right] +
          \displaystyle\sum_{i=1}^{k-1}\overline{\varepsilon}^{i+1}H_2(\delta_i)\displaystyle\prod_{m=1}^{i-1}\delta_m \nonumber \\
          &= \log(\overline{\delta}_0)\left[ \overline{\varepsilon}+\displaystyle\sum_{i=0}^{k-1}\overline{\varepsilon}^{i+2}\displaystyle\prod_{m=1}^{i}\delta_m \right]
          - \overline{\varepsilon}\log(\delta_0)+\displaystyle\sum_{i=1}^{k-1}\overline{\varepsilon}^{i+1}H_2(\delta_i)\displaystyle\prod_{m=1}^{i-1}\delta_m. \label{eq:last_align}
     \end{align}
     It is clear from \eqref{eq:last_align} that $\lim_{\delta_0\to 0^+}\frac{\partial N}{\partial \delta_0} D - N\frac{\partial D}{\partial\delta_0}=\infty$.

     Next we evaluate $\left( \frac{\partial N}{\partial \delta_0} D - N\frac{\partial D}{\partial\delta_0} \right)\Big|_{\delta_0=\frac{1}{2}}$:

     \begin{align}
       &\left( \frac{\partial N}{\partial \delta_0} D - N\frac{\partial D}{\partial\delta_0} \right)\Big|_{\delta_0=\frac{1}{2}} =
       -\left[ \overline{\varepsilon}+\displaystyle\sum_{i=0}^{k-1}\overline{\varepsilon}^{i+2}\displaystyle\prod_{m=1}^{i}\delta_m \right]
          + \overline{\varepsilon}+\displaystyle\sum_{i=1}^{k-1}\overline{\varepsilon}^{i+1}H_2(\delta_i)\displaystyle\prod_{m=1}^{i-1}\delta_m \nonumber \\
        &= \overline{\varepsilon}^2\left[ \left(H_2(\delta_1)-1 \right) + \overline{\varepsilon}\delta_1\left( H_2(\delta_2)-1 \right) + \ldots + \overline\varepsilon^{k-2}\displaystyle\prod_{m=1}^{k-2}\delta_m\left( H_2(\delta_{k-1})-1 \right) - \overline{\varepsilon}^{k-1}\displaystyle\prod_{m=1}^{k-1}\delta_m \right] \label{eq:delta_1/2 2}.
     \end{align}
     Note that all the summands are non-positive. It follows that:
     \begin{equation*}
       \left(H_2(\delta_1)-1 \right) + \overline{\varepsilon}\delta_1\left( H_2(\delta_2)-1 \right) + \ldots + \overline\varepsilon^{k-2}\displaystyle\prod_{m=1}^{k-2}\delta_m\left( H_2(\delta_{k-1})-1 \right) = 0
     \end{equation*}
     if and only if we set $\delta_1=\ldots=\delta_{k-1}=\frac{1}{2}$. However setting $\delta_1=\ldots=\delta_{k-1}=\frac{1}{2}$ we get:\\
     $- \overline{\varepsilon}^{k-1}\prod_{m=1}^{k-1}\delta_m < 0$. Thus $\left( \frac{\partial N}{\partial \delta_0} D - N\frac{\partial D}{\partial\delta_0} \right)\Big|_{\delta_0=\frac{1}{2}} < 0$.
    We now use the intermediate value theorem to prove the first part of the lemma.

     In the second part we want to show that the partial derivative $\frac{\partial R_\varepsilon\left(\delta_0, \ldots ,\delta_{k-1}\right)}{\partial \delta_0}$ is monotonic non-increasing in $\delta_0$. It is clear that $D^2$ is monotonic increasing in $\delta_0$ so we must prove that $\frac{\partial N}{\partial \delta_0} D - N\frac{\partial D}{\partial\delta_0}$ is non-increasing in $\delta_0$ to complete the proof. In order to achieve this goal we derive $\frac{\partial N}{\partial \delta_0} D - N\frac{\partial D}{\partial\delta_0}$ again with respect to $\delta_0$:

   \begin{equation*}
       \frac{\partial\left( \frac{\partial N}{\partial \delta_0} D - N\frac{\partial D}{\partial\delta_0} \right)}{\partial\delta_0} =
       \frac{-\left[ \overline{\varepsilon} + \displaystyle\sum_{i=0}^{k-1}\overline{\varepsilon}^{i+2}\displaystyle\prod_{m=1}^{i}\delta_m \right]}{\overline{\delta}_0}
       -\frac{\overline{\varepsilon}}{\delta_0}
     \end{equation*}
     This expression is clearly non-positive and that proves the lemma.
 \end{proof}

 We have shown that there is a unique $\vec{\delta}\in D_2$ that satisfies $\nabla R_\varepsilon(\vec{\delta}) = \vec{0}$. Now, all that remains is to prove that the suspicious point we worked so hard to find is, in fact, a local maximum of $R_\varepsilon(\delta_0,\ldots,\delta_{k-1})$ in the domain $D_1$. We already know that it is the only suspicious point in the interior of the domain, so we can safely say that the function gets its maximum value in that point or somewhere on the boundary. The final lemma will show that the function does not get its maximal value on the edge of the domain.

 \begin{lemma}\label{lemma:KKT}
   The maximum of $R_\varepsilon(\delta_0,\ldots,\delta_{k-1})$ does not occur on the boundary of the domain $D_1$.
 \end{lemma}
 \begin{proof}
   To prove this we will use the KKT conditions. First we will write the maximization problem in its standard form. Define $\vec{\delta} = (\delta_0,\ldots,\delta_{k-1})$ and the following constraint functions:
   \begin{align*}
     g_0(\vec{\delta}) = -\delta_0 \: &, \: \tilde{g}_0(\vec{\delta}) = \delta_0-1 \\
      \vdots \quad \quad & \quad \quad \vdots \\
      g_{k-1}(\vec{\delta}) = -\delta_{k-1} \: &, \: \tilde{g}_{k-1}(\vec{\delta}) = \delta_{k-1}-1.
   \end{align*}
   We want to maximize
   \begin{equation*}
     R_\varepsilon(\delta_0,\ldots,\delta_{k-1})
   \end{equation*}
   subject to
   \begin{equation*}
     g_i(\vec{\delta})\leq 0 \quad , \quad \tilde{g}_i(\vec{\delta})\leq 0 \quad\quad i=0,\ldots,k-1.
   \end{equation*}
   The KKT conditions tell us that if a point $\vec{\delta}^\ast$ is a local maximum then there exist constants $\mu_i$ and $\tilde{\mu}_i$ ($i=0,\ldots,k-1$) such that:
   \begin{equation}\label{eq:KKT_conditions}
     \nabla R_\varepsilon(\vec{\delta}^\ast) = \sum_{i=0}^{k-1}\mu_i\nabla g_i(\vec{\delta}^\ast) + \sum_{i=0}^{k-1}\tilde{\mu}_i\nabla \tilde{g}_i(\vec{\delta}^\ast)
   \end{equation}
   and
   \begin{equation*}
     \mu_i \geq 0 \quad , \quad \tilde{\mu}_i \geq 0 \quad \quad i=0,\ldots,k-1.
   \end{equation*}
   Let us assume, by contradiction, that $R_\varepsilon$ has a local maximum on the boundary of the domain $D_1$, and, specifically, that the local maximum is obtained for $\delta_0=0$. Since we assume that $g_0(\vec{\delta})=0$ we know that $\tilde{g}_0(\vec{\delta})=-1$ and so there is no need to address the inequality condition $\tilde{g}_0(\vec{\delta})\leq 0$. Eq. \eqref{eq:KKT_conditions} above gives us $k-1$ equalities. In this case the equality that we get from differentiating with regard to $\delta_0$ is:
   \begin{equation}\label{eq:dont_know}
     \frac{\frac{\partial N}{\partial\delta_0}D-N\frac{\partial D}{\partial\delta_0}}{D^2} = -\mu_0,
   \end{equation}
   where
   \begin{align*}
     \frac{\partial N}{\partial\delta_0}D =& \left(\log\left(\frac{\overline{\delta}_0}{\delta_0}\right) +\displaystyle\sum_{i=1}^{k-1}{\overline{\varepsilon}^{i+1}H_2(\delta_i)}\displaystyle\prod_{m=1}^{i-1}\delta_m\right)
     \left( 1+\displaystyle\sum_{i=0}^{k-1}\overline{\varepsilon}^{i+1}\displaystyle\prod_{m=0}^{i}\delta_i \right) \\
     N\frac{\partial D}{\partial\delta_0} =& \left( \displaystyle\sum_{i=0}^{k-1}{\overline{\varepsilon}^{i+1}H_2(\delta_i)}\displaystyle\prod_{m=1}^{i-1}\delta_m \right)
     \left( \displaystyle\sum_{i=0}^{k-1}\overline{\varepsilon}^{i+1}\displaystyle\prod_{m=1}^{i}\delta_i \right) \\
     D^2 =& \left(1+\displaystyle\sum_{i=0}^{k-1}\overline{\varepsilon}^{i+1}\displaystyle\prod_{m=0}^{i}\delta_i\right)^2.
   \end{align*}
   We have already shown in a previous lemma that the left hand side of Eq. \eqref{eq:dont_know} tends to $+\infty$ as $\delta_0\to 0^+$ and so we get a negative $\mu_0$ in violation of the KKT conditions. If we assume that a local maximum is obtained for $\delta_0=1$ we will arrive at a similar equation:
   \begin{equation}\label{eq:dont_know_2}
     \frac{\frac{\partial N}{\partial\delta_0}D-N\frac{\partial D}{\partial\delta_0}}{D^2} = \tilde{\mu}_0,
   \end{equation}
   and it is easy to see that the left hand side of Eq. \eqref{eq:dont_know_2} tends to $-\infty$ as $\delta_0\to 1^-$ so, again, we get a negative $\tilde{\mu}_0$. We can conclude that there is no local maximum of $R_{\varepsilon}$ on the boundary of $D_1$ where $\delta_0=0$ or $\delta_0=1$. In a similar way we can show that there is, in fact, no local maximum on any part of the boundary of $D_1$.
 \end{proof}

 We have proven that $R_\varepsilon(\delta_0,\ldots,\delta_{k-1})$ has one local maximum in the domain $D_1$ and that it satisfies $\frac{1}{2}>\delta_0\geq\delta_1\geq\ldots\geq\delta_{k-1}$. This proves that we can substitute $D_1$ for $D_2$ as the maximization domain and the maximal value will not change as a result.

\section{Lemma for the Converse}\label{sec:Lemma}
Prior to the proof we present some standard definitions of second order types. A second order type of a sequence $x^n\in\mathcal{X}^n$ is a probability distribution $\hat{P}_{x^n}^{(2)}\in\mathcal{P}_{n-1}(\mathcal{X}^2)$ defined as:
        \begin{equation*}
            \hat{P}_{x^n}^{(2)}(a,b)=\frac{N\left((a,b) \mid x^n\right)}{n-1},
        \end{equation*}
        for all $(a,b)\in(\mathcal{X}^2)$. Denote by $\mathcal{P}_n^{(2)}(\mathcal{X},c)$ the set of all possible second order types of sequences $x^n\in\mathcal{X}^n$ with $x_1=c$. The second order type of a sequence $x^n$ can be viewed as the joint empirical distribution of $x_1,x_2,\ldots,x_{n-1}$ and $x_2,x_3,\ldots,x_n$. For dummy random variables $X,Y$ representing such a second order type (i.e., $P_{X,Y}\in\mathcal{P}_n^{(2)}(\mathcal{X},c)$), we define a second order type class:
        \begin{equation*}
            T^{n,(2)}(P_{X,Y},c)=\{ x^n\in\mathcal{X}^n : \hat{P}_{x^n}^{(2)}=P_{X,Y} , x_1=c  \}.
        \end{equation*}
        We also define a second order $\varepsilon$-typical set with respect to a joint distribution $P_{X,Y}$ and $c$:
        \begin{align*}
            T_\varepsilon^{n,(2)}(P_{X,Y},c) = \Big\{ x^n\in\mathcal{X}^n : x_1=c, \; \forall (a,b)\in\mathcal{X}^2 \left| \hat{P}_{x^n}^{(2)}(a,b)-P_{X,Y}(a,b) \right|\leq \varepsilon \\
            \text{ and } P_{X,Y}(a,b) = 0 \Rightarrow \hat{P}_{x^n}^{(2)}(a,b)=0 \Big\}.
        \end{align*}
        A known result from \cite{method_of_types} is that for large enough $n$ there exists $\tau(\varepsilon)>0$ such that:
        \begin{equation*}
          2^{n\left(H(Y \mid X)-\tau(\varepsilon)\right)} \leq \left| T_\varepsilon^{n,(2)}(P_{X,Y},c) \right| \leq 2^{n\left(H(Y \mid X)+\tau(\varepsilon)\right)},
        \end{equation*}
        where $\lim_{\varepsilon \to 0} \tau(\varepsilon) = 0$.
    \begin{proof}[Proof of Lemma \ref{lemma:n-tuple_entropy}]
        We aim to show that there exists a single letter distribution such that the typical set induced by this distribution contains the typical set induced by the given n-tuple distribution. Since for a given distribution the size of its typical set is closely related to its entropy, that will suffice to prove the lemma.

      We emphasize that $n$ is fixed throughout the proof, so defined quantities are implicit functions of $n$. Consider an n-tuple constrained distribution $\tilde{P}_{Y^n}(y^n)=\mathbbm{1}_{\{y_1=1\}}\prod_{i=2}^{n}\tilde{P}^{(i)}_{Y_i \mid Y_{i-1}}(y_i \mid y_{i-1})$, as stated in the lemma. Set $\varepsilon=\frac{1}{(n-1)|\mathcal{X}|^n}$. Let $x^{nk} \in T^k_\varepsilon(\tilde{P}_{Y_n})$, meaning that the sequence $x^{nk}$ contains $k$ "letters" of the n-fold alphabet $\mathcal{X}^n$ and it is first order $\varepsilon$-typical with respect to the distribution $\tilde{P}_{Y^n}$. Since $x^{nk} \in T^k_\varepsilon(\tilde{P}_{Y^n})$ we have that $\forall x^n\in\mathcal{X}^n \; |\hat{P}_{x^{nk}}(x^n) - \tilde{P}_{Y^n}(x^n)|\leq\varepsilon$. Equivalently:
      \begin{equation}\label{eq:aux_lemma}
        \left|N(x^n \mid x^{nk}) - k\tilde{P}_{Y^n}(x^n)\right|\leq k\varepsilon.
      \end{equation}

      %The first step of the proof is to show that if $x^{nk}$ is second order typical with respect to the joint distribution $\tilde{P}_{X^n,Y^n}$ then $x^{n,2},\ldots,x^{n,k}$ is typical with respect to $P_{Y^n}$. Denote $\varepsilon'=\frac{\varepsilon}{|\mathcal{X}^n|}$, and let $x^{nk}\in T_{\varepsilon'}^{k,(2)}(\tilde{P}_{X^n,Y^n},c^n)$. We will now show that $x_{n+1}^{nk}\in T_\varepsilon^{k-1}(\tilde{P}_{Y^n})$. For any $b^n\in\mathcal{X}^n$:
%
%      \begin{align*}
%        \left| \hat{P}_{x_{n+1}^{nk}}(b^n) - \tilde{P}_{Y^n}(b^n) \right| &= \left| \sum_{a^n\in\mathcal{X}^n}\hat{P}^{(2)}_{x^{nk}}(a^n,b^n) - \sum_{a^n\in\mathcal{X}^n} \tilde{P}_{X^n,Y^n}(a^n,b^n)  \right| \\
%         &\leq \sum_{a^n\in\mathcal{X}^n} \left| \hat{P}^{(2)}_{x^{nk}}(a^n,b^n) - \tilde{P}_{X^n,Y^n}(a^n,b^n) \right| \\
%         &\leq |\mathcal{X}^n|\varepsilon' \\
%         &= \varepsilon
%      \end{align*}
%      This shows that
%      \begin{equation}\label{eq:N(a,b)}
%        \left| N(b^n|x_{n+1}^{nk}) - (k-1)\tilde{P}_{Y^n}(b^n) \right| \leq (k-1)\varepsilon
%      \end{equation}
%      Additionally, $\tilde{P}_{Y^n}(b^n) = 0 \implies \forall a^n\in\mathcal{X}^n, \; P_{X^n,Y^n}(a^n,b^n) = 0 \implies \hat{P}_{x^{nk}}^{(2)}(a^n,b^n) = 0$ which means that $\hat{P}_{x_{n+1}^{nk}}(b^n) = 0$. So indeed $x_{n+1}^{nk}\in T_\varepsilon^{k-1}(\tilde{P}_{Y^n})$

      Define a single letter joint distribution:
      \begin{equation}\label{eq:sigle_letter_dist}
        \tilde{P}_{X,Y}(a,b) = \sum_{x^n\in\mathcal{X}^n}\tilde{P}_{Y^n}(x^n)\frac{N\left( (a,b) \mid x^n \right)}{n-1}.
      \end{equation}
      We want to show that there exists $\delta>0$ such that $x^{nk}\in T_\delta^{nk,(2)}(\tilde{P}_{X,Y},1)$. In order to do that we need to calculate the empirical distribution of pairs of letters in $x^{nk}$. Each n-tuple $x^{n,i}$ contains $n-1$ pairs. The sequence $x^{nk}$ is made up of $k$ n-tuples, so there are an additional $k-1$ pairs that are not contained in a single n-tuple. In total there are $k(n-1)+k-1 = nk-1$ pairs of letters. For $(a,b)\in\mathcal{X}^2$ denote by $\eta(a,b)$ the number of times the pair $(a,b)$ appears in $x^{nk}$ and is not contained in a single n-tuple. Clearly $0\leq\eta(a,b)\leq k-1$. Now:
      \begin{align}
        \left| \hat{P}_{x^{nk}}^{(2)}(a,b) - \tilde{P}_{X,Y}(a,b) \right| &= \left| \frac{N\left( (a,b) \mid x^{nk} \right)}{nk-1} - \tilde{P}_{X,Y}(a,b) \right| \nonumber \\
         &= \left| \frac{1}{nk-1}\left[ \sum_{x^n\in\mathcal{X}^n}N(x^n \mid x^{nk})N\left( (a,b) \mid x^n \right) + \eta(a,b) \right] - \tilde{P}_{X,Y}(a,b) \right| \nonumber \\
         &= \Bigg| \frac{1}{nk-1}\Bigg[ \sum_{x^n\in\mathcal{X}^n} \left( N(x^n \mid x^{nk}) - k\tilde{P}_{Y^n}(x^n) \right)N\left( (a,b) \mid x^n \right) \nonumber \\
         &\quad \quad \quad \quad + \sum_{x^n\in\mathcal{X}^n} k\tilde{P}_{Y^n}(x^n) N\left( (a.b) \mid x^n \right) + N\left( (a,b) \mid c^n \right) + \eta(a,b) \Bigg] - \tilde{P}_{X,Y}(a,b) \Bigg| \nonumber \\
         &\stackrel{(a)}{=} \Bigg| \frac{1}{nk-1} \sum_{x^n\in\mathcal{X}^n} \left( N(x^n \mid x^{nk}) - k\tilde{P}_{Y^n}(x^n) \right)N\left( (a,b) \mid x^n \right) \nonumber \\
         &\quad \quad \quad + \frac{k(n-1)\tilde{P}_{X,Y}(a,b)}{nk-1} + \frac{\eta(a,b)}{nk-1} - \frac{(nk-1)\tilde{P}_{X,Y}(a,b)}{nk-1} \Bigg| \nonumber \\
         &\stackrel{(b)}{\leq} \left| \frac{|\mathcal{X}|^nk\varepsilon N\left( (a,b) \mid x^n \right)}{nk-1} \right| + \left| \frac{(1-k)\tilde{P}_{X,Y}(a,b)}{nk-1} \right| + \left| \frac{k-1}{nk-1} \right| \nonumber \\
         &\stackrel{(c)}{\leq} \left| \frac{k}{nk-1} \right| + \left| \frac{k-1}{nk-1} \right| + \left| \frac{k-1}{nk-1} \right|, \label{eq:last}
      \end{align}

      %\begin{align}
%        \hat{P}_{x^{nk}}^{(2)}(a,b) &= \frac{N\left( (a,b)|x^{nk} \right)}{nk-1} \\
%         &= \frac{1}{nk-1}\left[ \sum_{x^n\in\mathcal{X}^n}N(x^n|x^{nk})N\left( (a,b)|x^n \right) + \eta(a,b) \right] \\
%         &= \frac{1}{nk-1}\Bigg[ \sum_{x^n\in\mathcal{X}^n} \left( N(x^n|x^{nk}) - k\tilde{P}_{Y^n}(x^n) \right)N\left( (a,b)|x^n \right)\\
%         &\quad \quad \quad \quad + \sum_{x^n\in\mathcal{X}^n} k\tilde{P}_{Y^n}(x^n) N\left( (a.b)|x^n \right) + N\left( (a,b)|c^n \right) + \eta(a,b) \Bigg].
%      \end{align}
%
%       Recall from Eq. \eqref{eq:sigle_letter_dist} that $\sum_{x^n\in\mathcal{X}^n} \tilde{P}_{Y^n}(x^n) N\left( (a.b)|x^n \right)=(n-1)\tilde{P}_{X,Y}(a,b)$ and consider:
%      \begin{align*}
%        \left| \hat{P}_{x^{nk}}^{(2)}(a,b) - \tilde{P}_{X,Y}(a,b) \right| &= \Bigg| \frac{1}{nk-1} \sum_{x^n\in\mathcal{X}^n} \left( N(x^n|x^{nk}) - k\tilde{P}_{Y^n}(x^n) \right)N\left( (a,b)|x^n \right)\\
%        &\quad \quad \quad + \frac{k(n-1)\tilde{P}_{X,Y}(a,b)}{nk-1} + \frac{\eta(a,b)}{nk-1} - \frac{(nk-1)\tilde{P}_{X,Y}(a,b)}{nk-1} \Bigg|  \\
%         &\stackrel{(a)}{\leq} \left| \frac{|\mathcal{X}|^nk\varepsilon N\left( (a,b)|x^n \right)}{nk-1} \right| + \left| \frac{(1-k)\tilde{P}_{X,Y}(a,b)}{nk-1} \right| + \left| \frac{k-1}{nk-1} \right| \\
%         &\stackrel{(b)}{\leq} \left| \frac{k}{nk-1} \right| + \left| \frac{k-1}{nk-1} \right| + \left| \frac{k-1}{nk-1} \right|, \\
%      \end{align*}

      where
      \begin{enumerate}
        \item[(a)] Follows from Eq. \eqref{eq:sigle_letter_dist}: $\sum_{x^n\in\mathcal{X}^n} \tilde{P}_{Y^n}(x^n) N\left( (a.b) \mid x^n \right)=(n-1)\tilde{P}_{X,Y}(a,b)$.
        \item[(b)] Follows from the triangle inequality, Eq. \eqref{eq:aux_lemma} and $\eta(a,b)\leq k-1$.
        \item[(c)] Follows from the fact that $N\left( (a,b) \mid x^n \right) \leq n-1$, $\tilde{P}_{X,Y}(a,b) \leq 1$ and the definition of $\varepsilon$.
      \end{enumerate}
      From Eq. \eqref{eq:last}, for every $\xi>0$ there exist $k\in\mathbb{N}$ such that for any $(a,b)\in\mathcal{X}^2$
      \begin{equation*}
        \left| \hat{P}_{x^{nk}}^{(2)}(a,b) - P_{X,Y}(a,b) \right| \leq \frac{3}{n}+\xi.
      \end{equation*}
      Denote $\delta=\frac{3}{n}+\xi$. We have shown that $x^{nk} \in T^k_\varepsilon(\tilde{P}_{Y_n}) \implies x^{nk}\in T_\delta^{nk,(2)}(\tilde{P}_{X,Y},1)$. Thus, $\left| T^k_\varepsilon(\tilde{P}_{Y_n}) \right| \leq \left| T_\delta^{nk,(2)}(\tilde{P}_{X,Y},1) \right|$. For any pair $\tau_1(\varepsilon)>0$, $\tau_2(\delta)>0$ with $\lim_{\varepsilon\to 0}\tau_1(\varepsilon)=0$, $\lim_{\delta\to 0}\tau_2(\delta)=0$ there exist $K$ such that for all $k>K$:
      \begin{equation*}
        2^{k\left(H_{\tilde{P}_{Y^n}}(Y^n)-\tau_1(\varepsilon)\right)}\leq\left| T^k_\varepsilon(\tilde{P}_{Y^n}) \right| \leq \left| T_\delta^{nk,(2)}(\tilde{P}_{X,Y},1) \right| \leq 2^{nk\left( H_{\tilde{P}_{Y \mid X}}(Y \mid X) + \tau_2(\delta) \right)}.
      \end{equation*}
      So
      \begin{equation*}
        \frac{1}{n}\left( H_{\tilde{P}_{Y^n}}(Y^n)-\tau_1(\varepsilon) \right) \leq \left( H_{\tilde{P}_{Y \mid X}}(Y \mid X) + \tau_2(\delta) \right),
      \end{equation*}

      %This holds for any $\varepsilon>0$, $\delta > 0$ so letting $n\to\infty$ we can have $\varepsilon \to 0$, $\delta\to 0$ and conclude that:
%      \begin{equation*}
%        \lim_{n\to\infty}\frac{1}{n}H_{\tilde{P}_{Y^n}}(Y^n)\leq H_{\tilde{P}_{Y|X}}(Y|X),
%      \end{equation*}
      or, alternatively,
      \begin{equation*}
        H_{\tilde{P}_{Y^n}}(Y^n)\leq nH_{\tilde{P}_{Y \mid X}}(Y \mid X) + \zeta_n \quad , \quad \lim_{n\to\infty}\zeta_n = 0.
      \end{equation*}
      We can think of the single letter distribution $\tilde{P}_{Y \mid X}$ as $\tilde{P}_{Y_i,Y_{i-1}}(y_i \mid y_{i-1})$ and define $\tilde{Q}_{Y^n}(y^n)=\prod_{i=1}^{n}\tilde{P}_{Y_i,Y_{i-1}}(y_i \mid y_{i-1})$.
    \end{proof}

%************** THIS IS A BEGINNING OF AN APPENDIX REGARDING C_{(2,\infty)}^{\mathrm{fb}}(\frac{1}{2}) < C_{(2,\infty)}^{\mathrm{nc}}(\frac{1}{2})*******

%\section{Proof of $C_{(2,\infty)}^{\mathrm{fb}}(\frac{1}{2}) < C_{(2,\infty)}^{\mathrm{nc}}(\frac{1}{2})$}\label{sec:fb_less_nc}
%    We will use the following notation:
%    \begin{equation}\label{eq:h_func}
%      h_\varepsilon(\delta_0,\delta_1,\delta_2) = \frac{\overline{\varepsilon}\left(H_2(\delta_0)+\varepsilon H_2(\delta_1)+\varepsilon^2H_2(\delta_2)\right)}{1+\varepsilon+\varepsilon^2+2\overline{\varepsilon}(\delta_0+\varepsilon\delta_1+\varepsilon^2\delta_2)}.
%    \end{equation}
%    This function is defined on the compact domain $\{(\delta_0,\delta_1,\delta_2)\in \mathbb{R}^3 | 0\leq\delta_0,\delta_1,\delta_2\leq1\}$. By partially deriving we learn that $h_\varepsilon(\delta_0,\delta_1,\delta_2)$ has a single critical point and it is obtained for $\delta_0=\delta_1=\delta_2$. It can easily be shown that:
%    \begin{equation}\label{eq:concave1}
%      h_\varepsilon(p,p,p) = \frac{1}{2}\frac{H_2(p)}{\frac{1}{2(1-\varepsilon)}+p}
%    \end{equation}
%    It was proved in \cite{Sabag_BEC} that the RHS of \eqref{eq:concave1} is concave in $p$. The RHS of \eqref{eq:concave1} has

\section{Proofs for the Feedback Capacity of the $(1,2)$-RLL BEC}\label{sec:proofs_for_(1,2)-RLL}

\begin{proof}[\textbf{Proof of Lemma \ref{lemma:(1,2)_feasibility}}]
            The definition of the scheme clearly shows that if $Y_i=1$ then $X_{i+1}=0$, and if $Y_{i-1}=Y_i=0$ then $X_{i+1}=1$. All that remains is to prove that it is possible to assign $\delta$ of the unit interval to `$0$' in the case of consecutive erasures. We prove this using induction on the number of erasures, $n$. As a base case take $n=3$. Fig. \ref{figures:(1,2)_induction_base_case} shows possible partitions of the unit interval that comply with the definitions of the coding scheme.
            \begin{figure}
                \centering
                \psfrag{A}[c][][1]{`$0$'}
                \psfrag{B}[c][][1]{`$1$'}
                \psfrag{C}[c][][1]{$0$}
                \psfrag{D}[c][][1]{$1$}
                \psfrag{E}[c][][1]{$1-\delta$}
                \psfrag{F}[c][][1]{$\delta$}
                \psfrag{G}[c][][1]{$2\delta-1$}
                \psfrag{H}[c][l][1]{$y=0/?$}
                \psfrag{I}[c][][1]{$y=?$}
                \psfrag{J}[c][][1.3]{$l_4$}
                \psfrag{K}[c][][1.3]{$l_1$}
                \psfrag{L}[c][][1.3]{$l_2$}
                \includegraphics[scale=0.8]{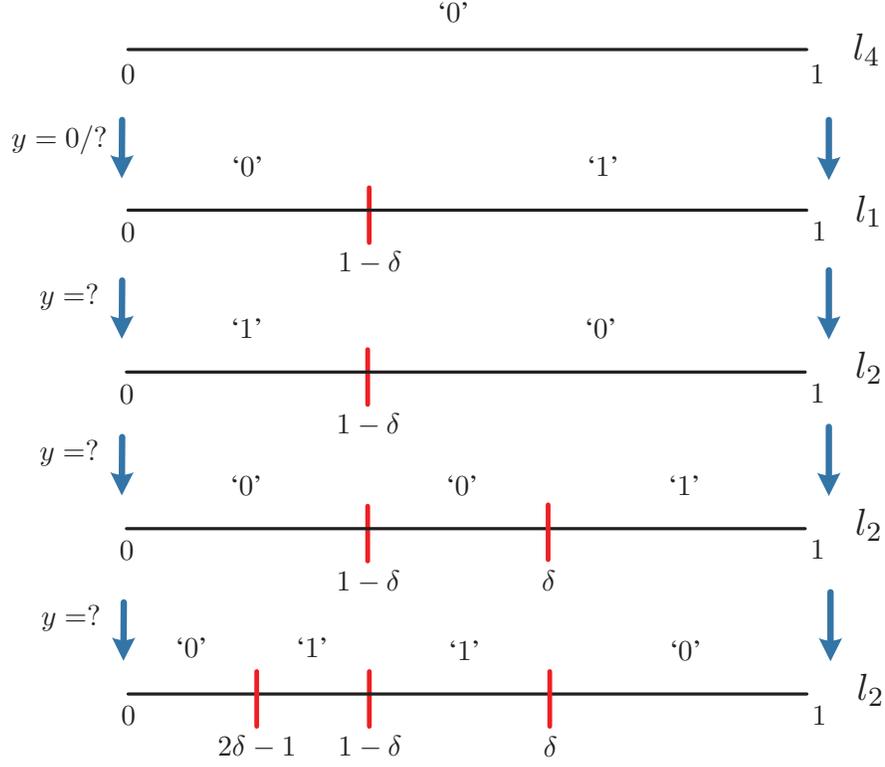}
                \caption{Induction base case: the initial node is 4 and all channel outputs are erasures. Each partition assigns the required amount of the unit interval to `$0$' and to `$1$'. The $(1,2)$-RLL input constraint is satisfied by these labelings: every `$1$' is followed by a `$0$' and no more than two consecutive `$0$'s are allowed.}\label{figures:(1,2)_induction_base_case}
            \end{figure}
            The last partition in Fig. \ref{figures:(1,2)_induction_base_case} assumes that $0\leq2\delta-1\leq1-\delta$. This means that $\frac{1}{2}\leq\delta\leq\frac{2}{3}$.

            For $n\geq4$ consecutive erasures consider the following:
            \begin{enumerate}
              \item Subintervals of $[0,1]$ which were labeled `$1$' during the previous channel use must now be labeled `$0$'. Additionally, the total length of these subintervals is $1-\delta$. This means that we must assign an additional $2\delta-1$ to `$0$'.
              \item Subintervals of $[0,1]$ which were labeled `$1$' two channel uses ago are now unconstrained. The total length of these subintervals is also $1-\delta$.
              \item The two sets of subintervals mentioned in items $1)$ and $2)$, above, are disjoint.
            \end{enumerate}
            Given that $\delta\leq\frac{2}{3}$ we can assign $2\delta-1$ of the unconstrained subintervals to `$0$'.

            We calculate the rate achieved by this scheme using the same method shown in Section \ref{sec:coding_scheme}. In the case of this coding scheme, and after exchanging $\delta$ with $1-\delta$, we reach the following rate:
        \begin{equation}\label{}
            R = \frac{H_2(\delta)} {\frac{1}{1-\varepsilon}+\overline{\varepsilon}+\delta}.
        \end{equation}
        \end{proof}

            \begin{proof}[\textbf{Proof of Lemma \ref{lemma:(1,2)_bounds_coincide}}]
                Eq. \eqref{(eq:1,2)_1st_upper_bound} contains an upper bound to $C^{\mathrm{fb}}_{(1,2)}(\varepsilon)$. By partially deriving the RHS of Eq. \eqref{(eq:1,2)_1st_upper_bound}, it can be shown that the maximum is attained for $\delta_2=1-\delta_1$. Substituting this relation into Eq. \eqref{(eq:1,2)_1st_upper_bound} gives:
            \begin{equation}\label{(1,2)_2nd_upper_bound}
                C_{(1,2)}^{\mathrm{fb}}\leq \max_{0\leq\delta\leq1} \frac{H_2(\delta)} {\frac{1}{1-\varepsilon}+\overline{\varepsilon}+\delta}.
            \end{equation}
            Since $H_2(x)=H_2(1-x)$, it is clear that the maximum is in $0\leq\delta\leq\frac{1}{2}$. The derivative of Eq. \eqref{(1,2)_2nd_upper_bound} is equal to zero only when
            \begin{equation}\label{eq:M_2}
                (1-\delta)^{\frac{1}{1-\varepsilon}+\overline{\varepsilon}+1} = \delta^{\frac{1}{1-\varepsilon}+\overline{\varepsilon}}.
            \end{equation}
            The LHS and RHS of Eq. \eqref{eq:M_2} are, respectively, monotonic decreasing and monotonic increasing functions of $\delta$. In order for the maximizing $\delta$ to be at least $\frac{1}{3}$, we need to prove that for any $\varepsilon$
            \begin{equation}\label{eq:M_3}
                \left( \frac{2}{3} \right)^{\frac{1}{1-\varepsilon}+\overline{\varepsilon}+1} \geq \left( \frac{1}{3} \right)^{\frac{1}{1-\varepsilon}+\overline{\varepsilon}},
            \end{equation}
            which simplifies to
            \begin{equation}\label{eq:M_4}
                2^{\frac{1}{1-\varepsilon}+\overline{\varepsilon}+1} \geq 3.
            \end{equation}
            Since the power increases in $\varepsilon$, it is sufficient to check that Eq. \eqref{eq:M_4} holds for $\varepsilon = 0$, and indeed $2^3=8\geq3$.
            \end{proof}
\end{appendices}

\end{document}